\documentclass[12pt, draftclsnofoot, onecolumn]{IEEEtran}
\usepackage{amsmath,amssymb,amsthm,mathrsfs,bbm, comment, color, xcolor}
\usepackage{epsfig,epsf,subcaption,graphicx,graphics, algorithm, algorithmic}

\excludecomment{double}
\includecomment{single}

\theoremstyle{plain}
\newtheorem{thm}{Theorem}
\newtheorem{prop}{Proposition}
\newtheorem{lem}{Lemma}
\theoremstyle{remark}

\graphicspath{./figs/}

\usepackage{xifthen}

\newtheorem{theorem}{Theorem}[section]

\newtheorem{definition}{Definition}[section]
\newtheorem{remark}{Remark}[section]

\newcommand{\SNR}{\mathsf{SNR}}

\newcommand{\aaaa}{\mathrm{(a)}}
\newcommand{\bbbb}{\mathrm{(b)}}

\newcommand{\E}[1]{\mathbb{E} \left[ #1 \right]}
\newcommand{\Prob}[1]{\mathbb{P} \left( #1 \right)}
\newcommand{\lp}{\left(}
\newcommand{\rp}{\right)}
\newcommand{\lb}{\left[}
\newcommand{\rb}{\right]}
\newcommand{\lbp}{\left\{}
\newcommand{\rbp}{\right\}}

\newcommand{\what}{\widehat}
\newcommand{\wtild}{\widetilde}

\newcommand{\msf}{\mathsf}

\newcommand{\cgauss}[2]{\mathcal{CN}(#1,#2)}

\newcommand{\diag}{\text{diag}}

\DeclareMathOperator*{\argmax}{arg\,max}

\begin{single}

\end{single}

\allowdisplaybreaks

\title{Enhancing Multiuser MIMO Through Opportunistic D2D Cooperation}
\author{
\IEEEauthorblockN{Can Karakus} ~~
\and
~~\IEEEauthorblockN{Suhas Diggavi} 
\thanks{The authors are with the Department of Electrical Engineering, University of California, Los Angeles. E-mail: \{karakus, suhasdiggavi\}@ucla.edu. This work was supported in part by NSF grants \#1314937 and \#1514531 and a gift from Intel.}}
\begin{document}

\maketitle

\begin{abstract}
We propose a cellular architecture that combines multiuser MIMO (MU-MIMO) downlink with
opportunistic use of unlicensed ISM bands to establish
device-to-device (D2D) cooperation. The architecture consists of a physical-layer cooperation scheme
based on forming downlink virtual MIMO channels through D2D relaying, and a novel resource allocation strategy for such D2D-enabled networks. We prove the approximate optimality of the
physical-layer scheme, and demonstrate that such cooperation boosts the effective $\SNR$ of the weakest user in the system, especially in the many-user regime, due to multiuser
diversity. To harness this physical-layer scheme, we formulate the cooperative user
scheduling and relay selection problem using the network utility
maximization framework. For such a cooperative network, we propose a
novel utility metric that jointly captures fairness in throughput and
the cost of relaying in the system. We propose a joint user scheduling
and relay selection algorithm, which we prove to be asymptotically
optimal. We study the architecture through system-level simulations over a wide range of scenarios. The highlight of these simulations is an approximately $6$x improvement
in data rate for cell-edge (bottom fifth-percentile) users (over the
state-of-the-art SU-MIMO) while still improving the overall throughput, and
taking into account various system constraints.
\end{abstract}

\section{Introduction}\label{sec:intro}
One of the biggest challenges in wireless networks is to
provide uniform connectivity experience throughout the service
area. The problem is especially difficult at the cell-edge, where users with unfavorable
channel conditions need to receive reliable and high-rate
communications. One of the ambitious visions of 5G network
design is to achieve $10$x reduction in data rate variability in the cell \cite{OttHimayat_16} (over existing 4G single-user MIMO OFDM architecture with proportional fair scheduling), without sacrificing the overall sum throughput in the system. In
this paper, we propose and study a solution that, realistic simulations
indicate, can give up to approximately $6$x improvement in data rate for
cell-edge (bottom fifth-percentile) users while still improving the overall throughput under various system constraints.

The proposed solution is centered around opportunistically using the
unlicensed band through device-to-device (D2D) cooperation to improve
the performance of the licensed multiple-antenna downlink
transmission. This solution can be enabled without the
presence of any WiFi hotspots, or other data off-loading
mechanisms. The main idea is an architecture where a multiple-antenna
downlink channel is enhanced through out-of-band D2D relaying to provide
multiple versions of the downlink channel outputs, forming virtual MIMO links, which is then
opportunistically harnessed through scheduling algorithms designed for this architecture. 

This architecture is predicated on two opposing developments. The first is
that infrastructure is becoming more powerful, with the use of a growing number of
multiple antennas through massive MIMO for 5G. The other development is on the user equipment (UE)
side, with mobile devices becoming more powerful, both in terms of spectrum access and computational power. Most of the mobile devices currently in
widespread use can access multiple bands over the ISM spectrum,
including the 2.4GHz and 5GHz bands. 
Furthermore, dense clusters of users 
constitute a challenging scenario for increasing capacity through massive MIMO, which
is precisely the scenario where D2D cooperation is the most useful, since the D2D links
are much stronger. 

The main technical question involving the architecture is that of how and when to enable such D2D
links in a network with many users to boost the cell-edge gains. Our analysis, which uses the network utility maximization framework, leads to an optimal resource allocation algorithm for scheduling these links in a centralized manner, while accounting for system constraints
such as limited network state knowlede at the base station; uncoordinated interference over the unlicensed band; fairness in throughput and fairness in the amount of relaying performed by users.
Extensive simulations based on 3GPP channel
models demonstrate that the proposed architecture combined with our resource allocation algorithm can yield up to approximately
$6$x throughput gain for the bottom fifth-percentile of users in the
network and up to approximately $4$x gain for median users over the state-of-the-art single-user MIMO (SU-MIMO) currently implemented in LTE systems, without degrading the throughput of the high-end users.

Since the architecture relies on opportunistically using the unlicensed ISM bands, an important question is how the D2D transmissions would affect other wireless technologies using the unlicensed bands, such as WiFi. As a co-existence mechanism, one can consider strategies similar to LTE-U \cite{ZhangWang_15}: a user can search for an available (unused) channel within the unlicensed band to use for D2D cooperation. If none exists, the user can either declare itself unavailable for D2D cooperation, or transmit only for a short duty cycle. We study the effect of a simplified co-existence mechanism that does the former through simulations, and find that the throughput loss in WiFi users is small compared with the gains in the cell-edge users, since the fraction of time D2D transmission is required from a given user is small.


The main technical contributions of the work can be summarized as follows.
\begin{itemize}

\item We analyze a physical-layer scheme based on compress-and-forward relaying and MIMO Tx/Rx processing that approximately achieves (within 2 bits/s/Hz) the capacity of two-user downlink channel with D2D
cooperation (Section~\ref{sec:cooperation}), and describe how the scheme can be extended to MU-MIMO (Section~\ref{sec:mumimo}). We characterize the gains in terms of cell-edge
$\SNR$-scaling due to D2D cooperation for a specific model of clustered networks (Section~\ref{sec:scaling}).

\item We formulate the problem of allocating such D2D links for cooperation within the utility maximization
framework (Section~\ref{sec:formulation}). Since the existing cross-layer design tools are not directly applicable in our scenario when D2D transmission conflicts are taken into account, we propose a novel scheduling policy for such D2D-enabled networks that takes into account such conflicts (Section~\ref{sec:optimal}). The policy consists of an extension of the single-user scheduling algorithm of \cite{TsibonisGeorgiadis_05} to the cooperative MU-MIMO scenario with incomplete network state knowledge, and a novel flow control component based on an explicit characterization of an inner bound on the stability region of the system. The proposed algorithm is shown to be optimal with respect to this inner bound on the stability region. 
We also introduce a novel class of utility functions for cooperative downlink communication, which incorporates the cost of cooperation and leads to desirable fairness properties (Section~\ref{sec:utility}). 

\item We present an extensive simulation study using 3GPP specifications to study the performance of the
proposed architecture (Section~\ref{sec:sim}). The main results include {\sf
(i)} a throughput gain ranging from $4.3$x up to $6.3$x (depending on system constraints, channel estimation accuracy etc.) for the users in bottom
fifth-percentile for MU-MIMO with D2D cooperation versus the state-of-the-art SU-MIMO, without degrading the throughput of the stronger users, {\sf (ii)} a throughput gain ranging from $3.7$x up to $4.9$x for the bottom
fifth-percentile users versus non-cooperative MU-MIMO without degrading throughput of stronger users, 
{\sf (iii)} A reduction of more than $50\%$ in the relaying load in the network through the use of novel utility functions, while still giving gains close to proportional fair case, {\sf (iv)} A basic study of an architecture wherein D2D cooperation coexists (and interferes) with WiFi in the network via a simple co-existence mechanism where cooperation is disabled within WiFi range, where it is shown that despite the residual interference, the throughput loss in WiFi users is small (10\% for median user) compared with the gains in the cell-edge users (130\% for fifth-percentile user), since the fraction of time D2D transmission is required from a given user is small (in the simulation 80\% of users performed relaying less than 10\% of the time).
\end{itemize}

{\noindent \bf Related work:} 
The relevant literature can be broadly classified into three areas: {\sf (i)} cooperative cellular communications; {\sf (ii)} dynamic downlink scheduling; {\sf (iii)} D2D in cellular communications; each of which we will summarize next.

In cooperative cellular communications, the idea is to allow users overhearing transmissions to perform relaying to increase spatial diversity and minimize outage probability. This line of work (for instance, \cite{SendonarisErkip_03, NosratiniaHunter_04, LiuTao_06}, and the references therein) typically focuses on uplink and in-band cooperation, where users that overhear other users' transmission over the licensed band relay their version to the base station. In contrast, we focus on downlink communication and out-of-band cooperation, where users perform relaying for each other's downlink traffic by \emph{opportunistically} using the unlicensed band. As will be seen, the use of orthogonal bands for cooperation can significantly simplify coding schemes. 

There is also a large literature in cellular downlink scheduling. Some of these works focus on scaling behavior of the achievable rate under various scheduling schemes \cite{SharifHassibi_05, YooGoldsmith_06}, some focus on the low-complexity algorithms \cite{DimicSidiropoulos_05}, while some others also account for fairness and various system constraints using the cross-layer optimization approach \cite{LiuChong_01, TsibonisGeorgiadis_05, LinShroff_06, GeorgiadisNeely_06, ShiraniCaire_10}. While our work uses the cross-layer optimization paradigm as well, none of the proposed resource allocation algorithms directly applicable to our cooperative scenario, since we consider an architecture where the broadcast nature of the wireless medium is explicitly used at the physical-layer, precluding an abstraction into isolated bit pipes in upper layers, which is a prevalent model in existing works on cross-layer optimization.

Embedding D2D communication in cellular network has also received considerable attention in the past (see \cite{AsadiWang_14} for a comprehensive survey). A majority of these works (\emph{e.g.}, \cite{DopplerRinne_09, LiLeiGao_12, WuTavildar_13}) focus on direct proximal communication between devices, where one device directly transmits a message for another over the licensed band, skipping infrastructure nodes. This type of proximal D2D communication also has been part of the 4G LTE-Advanced standard \cite{LiuKato_15}. The main focus in this line of work is to do resource allocation and interference management across D2D and/or uplink/downlink message flows. In contrast, we focus on D2D \emph{cooperation} to aid downlink communication, which is the use of physical D2D transmissions to assist downlink message flows intended for other devices. This can be considered as a new way the D2D capability can be used in the next-generation 5G networks, in addition to the existing proximal communication in 4G. Considering the fact that the volume of downlink traffic far exceeds the volume of proximal D2D communication traffic, the cooperation architecture has the potential to exploit the D2D capability to a much higher degree. This is also in line with one of the envisioned goals in 5G, which is to enable multihop communication in cellular networks \cite{ChenZhao_14}. 

Conceptually, the most relevant work in the literature to our problem is the one in \cite{AsadiMancuso_13_2}, where the authors propose an architecture where users form clusters through the use of unlicensed bands, and all communication with the base station is performed through the cluster head. In another line of work \cite{WangRengarajan_13}, the authors suggest using out-of-band D2D for traffic spreading, where a user performs sends request and receives downlink content on behalf of another user, in a base-station transparent manner. In both works, the authors numerically demonstrate various throughput, fairness and energy-efficiency benefits of D2D. 
In contrast to these works, our physical-layer scheme is not based on routing; it explicitly uses the direct link from the base station to the destination user in addition to the relay links. We also consider a much more general scheduling algorithm based on utility optimization and dynamic user pairing, while accounting for fairness and cooperation cost.

{\noindent \bf Paper outline: } In
Section~\ref{sec:model}, we present our model and proposed architecture. In Section~\ref{sec:phy},
we present the physical-layer cooperation scheme, prove its
approximate optimality, describe its extension to MU-MIMO,
and study the scaling behavior of the minimum effective $\SNR$ in the
network. In
Section~\ref{sec:scheduling}, we formulate the downlink cooperative
scheduling problem within the utility optimization framework and present
our scheduling algorithm, along with the proposed cooperative utility metric, and in Section~\ref{sec:sim}, we present our simulation results.

\section{Model and Overview of the Architecture}\label{sec:model}
\subsection{Overview of the Architecture}

Consider a single cell in a multi-cell downlink cellular system\footnote{Since the base stations are uncoordinated, for the purposes of designing a scheduling algorithm, it is sufficient to consider a single cell in isolation. We will consider the multi-cell system in Section V for evaluation purposes.} with a base station equipped with $M$ antennas, and a set $\mathcal{N}$ of single-antenna users, where $\left| \mathcal{N}\right|=n$. An example operation is depicted in Figure~\ref{fig:schedule}. 
\begin{figure}
\centering
  \includegraphics[scale=0.8]{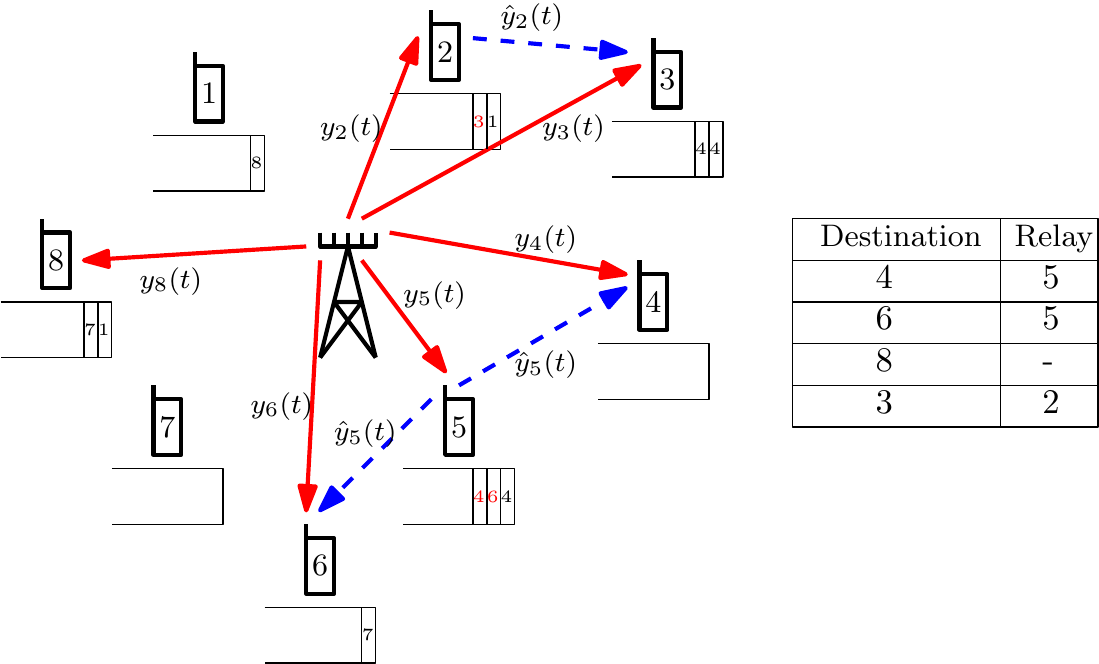}
  \caption{An example scheduling decision made by the base station, where the table reflects the selected active set. The red arrows denote the corresponding downlink transmissions, all taking place throughout frame $t$, and the dashed blue arrows represent scheduled side-channel transmissions, taking place at a later time, determined by multiple-access protocol $\mathcal{I}$. Once the active set is selected, the required side-channel transmissions are queued at the users (the transmissions scheduled in frame $t$ are highlighted in red). In this example, user $5$ is selected to share a function of its channel output to relay for users $4$ and $6$; user $2$ is selected to relay for user $3$; and user $8$ is scheduled without any relays.}
  \label{fig:schedule}
\end{figure}
We assume slotted time, with $m$ representing the physical-layer time index. A frame, indexed by $t$, is defined as $T$ consecutive discrete time slots\footnote{We will use square brackets to denote physical-layer time indices, and round brackets for frame indices.}. We will use the notation $m \sqsubset t$ to mean that the physical-layer slot $m$ lies within the frame $t$, \emph{i.e.}, $(t-1)T < m \leq tT$.

In the proposed architecture, the base station selects an active set $\mathcal{A}(t) \subseteq \mathcal{N}^2$ for each frame $t$, which consists of pairs $(i,j)$ of users, where the first index $i$ refers to the destination node scheduled for data, and the second index $j$ refers to user assigned as a relay for user $i$. We define $(i,i)$ to represent the case where user $i$ is scheduled with no relay assigned. Note that a user can be designated as a relay for a stream and a destination for another stream simultaneously, as exemplified in Figure~\ref{fig:schedule}. It is also possible within this framework to assign multiple relays to the same destination by having $(i,j),(i,k) \in \mathcal{A}(t)$. We define $A_{ij}(t)=1$ if $(i,j) \in \mathcal{A}(t)$, and $A_{ij}(t)=0$ otherwise.

Once the selection $\mathcal{A}(t)$ is made, the base station transmits a sequence of vectors $\mathbf{x}[m] \in \mathbb{C}^{M}$, $m=1,\dots,T$, over $M$ antennas and $T$ time slots of the frame $t$. The channel output $y_i[m]$ at user $i$ is given by
\begin{align}
y_i[m] = \mathbf{h}_i^*(t) \mathbf{x}[m] + w[m], \label{eq:main_channel}
\end{align}
for $m \sqsubset t$, where $\mathbf{h}_i (t) \in \mathbb{C}^M$ is the time-variant complex channel vector of user $j$ at frame $t$ (note that we are assuming that channel stays constant within a frame, but can arbitrarily vary over time slots), $\mathbf{x}[m]$ is the input vector to the channel at time $m$, and $w[m] \sim \cgauss{0}{1}$ is the circularly symmetric complex white Gaussian noise process. We assume an average power constraint $\frac{1}{T}\sum_{m=1}^T \mathbf{tr} \lp  \mathbf{x}[m]  \mathbf{x}^*[m]\rp \leq 1$, and define $H(t) := \lbp  \mathbf{h}_i(t)\rbp_i$.

If user $j$ is assigned as a relay for user $i$ at frame $t$, a transmission from user $j$ to $i$ is queued at user $j$, to be transmitted at a later frame $\tau > t$. At frame $\tau$, user $j$ transmits the sequence $x_j[m] \in \mathbb{C}$, $m \sqsubset \tau$, which is a deterministic function of the receptions corresponding to earlier frame $t$, \emph{i.e.}, $y_j[\wtild m]$ for $\wtild m \sqsubset t$. User $i$ performs decoding by combining its own channel outputs $y_i[m]$, $m \sqsubset t$, with the receptions from $j$, $\bar y_j [\wtild m]$, $\wtild m \sqsubset \tau$, which is a function of $y_j[m]$, $m \sqsubset t$ (the specific D2D link model generating $\bar y_j [\wtild m]$ will be discussed later). Note that user $i$ can combine receptions corresponding to multiple frames to decode.

We will specify the details of the model and formulate the specific mathematical problem.

\begin{single}
\begin{table}\centering
\scriptsize 
\caption{Notation for variables corresponding to the D2D link $(i,j)$}
\begin{tabular}{| c | l || c | l | }
\hline
{\bf Notation} & {\bf Explanation} & {\bf Notation} & {\bf Explanation} \\
\hline
$g_{ij}$ & D2D channel gain  &$Q_{ij}$ & State of the queue at relay $j$ for destination $i$  \\
$\phi_{ij}, \Phi$ & Path-loss factor(s)   & $\mu_{ij}$ & Binary service process (transmission permission indicator) for the queue $Q_{ij}$   \\
$\zeta_{ij}, Z$ & Fading parameter(s) & $A_{ij}$ & Binary arrival process (D2D link scheduling indicator) for the queue $Q_{ij}$ \\
$B_{ij}$ & D2D link availability indicator & $J_{ij}$  & D2D interference indicator  \\
$C_{ij}$ & The capacity of the D2D link  & $\beta_{ij}$ & Arrival rate to the queue $Q_{ij}$ \\
\hline
\end{tabular} 
\label{tb:notation}
\end{table}
\end{single}
\begin{double}
\begin{table}\centering
\scriptsize 
\caption{Notation for variables corresponding to the D2D link $(i,j)$}
\begin{tabular}{| c | l | }
\hline
{\bf Notation} & {\bf Explanation} \\
\hline
$g_{ij}$ & D2D channel gain  \\
$\phi_{ij}$ & Path-loss factor   \\ 
$\zeta_{ij}$ & Fading parameter \\
$B_{ij}$ & D2D link availability indicator \\
$C_{ij}$ & The capacity of the D2D link  \\
$Q_{ij}$ & State of the queue at relay $j$ for destination $i$  \\
$\mu_{ij}$ & Binary service process for the queue $Q_{ij}$   \\
$A_{ij}$ & Binary arrival process for the queue $Q_{ij}$ \\
$J_{ij}$  & D2D interference indicator  \\
 $\beta_{ij}$ & Arrival rate to the queue $Q_{ij}$ \\
\hline
\end{tabular} 
\label{tb:notation}
\end{table}
\end{double}

\subsection{D2D Link Model and Conflict Graph}

For any pair $(i,j) \in \mathcal{N}^2, i \neq j$, the time-variant channel gain is given by $g_{ij}(t) = \sqrt{\phi_{ij}} \zeta_{ij} (t)$, where $\phi_{ij} \in \mathbb{R}$ is the path loss component, and $\zeta_{ij} (t) \sim \mathcal{CN}(0,1)$ is the fading component for the pair $(i,j)$, i.i.d. across MAC layer slots. We assume reciprocal side-channels, \emph{i.e.}, $g_{ij}(t)=g_{ji}(t)$, and define $Z(t) := \lbp {\zeta}_{ij}(t)\rbp_{i,j}$ and $\Phi := \lbp  \phi_{ij}\rbp_{i,j}$. 

We define $B_{ij} (t)$ as an i.i.d. $Bernoulli(p_{ij})$ process for each $(i,j) \in \mathcal{N}^2, i \neq j$, representing whether or not the link $(i,j)$ is available at frame $t$. This models unavailability due to external transmissions {(\emph{e.g.}, WiFi access points, or another application on the same device attempting to use WiFi etc.)} in the same unlicensed band. The realization of $B_{ij}(t)$ is known at the users strictly causally (at frame $t+1$), and unknown at the base station. We define $B(t) := \{ B_{ij} (t)\}_{i\neq j}$.

Define the connectivity graph $\mathcal{G}=\lp \mathcal{V}, \mathcal{E}\rp$ such that $\mathcal{V} = \mathcal{N}$, and $\mathcal{E}$ is such that $\lp i,j\rp \in \mathcal{E}$ if $i = j$ or $\phi_{ij} > \theta$ for some threshold $\theta > 0$ (\emph{e.g.}, noise level). We further define the conflict graph $\mathcal{G}_c = \lp \mathcal{V}_c, \mathcal{E}_c\rp$ such that
\begin{single}
\begin{align}
\mathcal{V}_c &:= \lbp (i,j) \in [n]^2: i \neq j\rbp \label{conflict_graph}\\
\mathcal{E}_c &:= \lbp \lp (i,j), (k, \ell)\rp: (i,j) \neq (k, \ell) \text{ and } \lp \lp i, \ell\rp \in \mathcal{E} \text{ or }  \lp j, k\rp \in \mathcal{E} \rp\rbp. \notag
\end{align}
\end{single}
\begin{double}
\begin{align}
\mathcal{V}_c &:= \lbp (i,j) \in [n]^2: i \neq j\rbp \label{conflict_graph}\\
\mathcal{E}_c &:= \{ \lp (i,j), (k, \ell)\rp: (i,j) \neq (k, \ell) \text{ and }\notag\\
&\qquad  \lp \lp i, \ell\rp \in \mathcal{E} \text{ or }  \lp j, k\rp \in \mathcal{E} \rp\}. \notag
\end{align}
\end{double}
The conflict graph represents the pairs of D2D transmissions $(i,j), (k,\ell)$ that are not allowed to simultaneously occur due to interference\footnote{The interference model that induces the conflict graph $\mathcal{G}_c$ as defined in \eqref{conflict_graph} is similar to the two-hop interference model of \cite{Arikan_84}, but also takes into account the directionality of the transmission}. Given these definitions, the channel from user $j \in \mathcal{N}$ to user $i \in \mathcal{N} - \lbp j\rbp$ is modeled by
\begin{align*}
\bar y_i[m] = B_{ij}(t)J_{ij}(t)\lp g_{ij}(t)x_j[m] + \bar w_i [m]\rp
\end{align*}
for $m \sqsubset t$, where 
\begin{single}
\begin{align*}
J_{ij} (t) &= \left\{
\begin{array}{ll}
0, & \text{if $\exists (k,\ell)$ s.t. $\lp (i, j), (k,\ell)\rp \in \mathcal{E}_c$ and $\| x_{\ell}[\tilde m]\|^2 > 0$ for some $\tilde m \sqsubset t$} \\
1, & \text{otherwise}
\end{array}
\right.,
\end{align*}
\end{single}
\begin{double}
\begin{align*}
J_{ij} (t) &= \left\{
\begin{array}{ll}
0, & \begin{array}{l} \text{if $\exists (k,\ell)$ s.t. $\lp (i, j), (k,\ell)\rp \in \mathcal{E}_c$ and } \\ \text{$\| x_{\ell}[\tilde m]\|^2 > 0$ for some $\tilde m \sqsubset t$} \end{array} \\
1, & \text{otherwise}
\end{array}
\right.,
\end{align*}
\end{double}
which captures interference between conflicting D2D transmissions, and $\bar w_i[m] \sim \cgauss{0}{1}$ is the complex white Gaussian noise process. We assume an average power constraint $\frac{1}{T}\sum_{m=1}^T \| x_j[m]\|^2 \leq 1$, absorbing the input power into the channel gain. The capacity of the D2D link $(i,j)$ at time $t$ (assuming it is available) is given by $C_{ij} (t) := \log\lp 1+\| g_{ij}(t)\|^2\rp$. We assume the base station has knowledge of the average $\SNR$, \emph{i.e.}, the path-loss component $\phi_{ij}$ for each $(i,j)$ pair, but has no knowledge of the fading realization $\zeta_{ij}(t)$.

\subsection{D2D Transmission Queues}

We assume that each user $j \in \mathcal{N}$ maintains $\lp n-1\rp$ queues, whose states are given by $Q_{ij} (t)$, $i \in \mathcal{N}-\lbp j \rbp$, each representing the number of slots of transmission\footnote{Note that $Q_{ij} (t)$ does not represent the number of \emph{bits} to be transmitted, but the number of \emph{slots of transmission}. This is because the reception of relay does not directly translate into information bits, but is rather a \emph{refinement} of the reception of the destination node.} to be delivered to node $i$. We assume the queue states evolve according to
\begin{align}
Q_{ij}(t+1) = \lp Q_{ij}(t) - B_{ij}(t)J_{ij}(t)\mu_{ij} (t)\rp^+ + A_{ij}(t), \label{eq:queue_update}
\end{align}
where $\mu_{ij} (t)$ is a binary process that is induced by the multiple-access protocol $\mathcal{I}$ used by the nodes, indicating whether or not the flow $(i,j)$ is granted permission for transmission at frame $t$. The protocol $\mathcal{I}$ is a mapping from the current queue states $\lbp Q_{ij}(t)\rbp_{i \neq j}$ and the D2D interference structure $\lbp J_{ij}(t)\rbp_{i \neq j}$ to the binary service processes $\lbp \mu_{ij}(t)\rbp_{i \neq j}$. 

We define the average arrival rates as $\beta_{ij} (t) := \frac{1}{t}\sum_{\tau = 1}^t A_{ij}(t)$, and $\beta_{ij} := \limsup _{t \to \infty} \beta_{ij} (t)$. For a given vector of arrival rates $\beta := \lbp \beta_{i,j} \rbp_{i \neq j}$, the system is said to be \emph{stable} if the average queue sizes are bounded, \emph{i.e.}, for all $\lp i,j \rp$, $\limsup_{t \to \infty} \mathbb{E}\lb Q_{ij} (t)\rb < \infty$. The set of arrival-rate vectors $\beta$ for which there exists service processes $\lbp \mu_{ij}(t)\rbp_{i\neq j}$ such that the system is stable is called the \emph{stability region} of the queueing system, and will be denoted by $\Lambda$. Note that the arrival rates need to remain in the stability region in order to ensure that the D2D transmissions eventually occur with a finite delay. Within the scope of this paper, we do not focus on the details of $\mathcal{I}$, and simply assume that the nodes implement a protocol $\mathcal{I}$ that achieves the stability region $\Lambda$, \emph{i.e.}, if the arrival rates $\beta \in\Lambda$, protocol $\mathcal{I}$ can find a schedule for D2D transmissions such that each transmission is successfully delivered with finite delay\footnote{One can design such a protocol by having the nodes coordinate with the base station to circumvent the hidden terminal problem, and then use any of the existing stability-region-achieving distributed scheduling algorithms, \emph{e.g.}, \cite{TassiulasEphremides_92, LibinWalrand_10, ModianoShah_06}}.

\subsection{Problem Formulation}
If the vector of arrival rates $\beta \in \Lambda$, we can assume that a noiseless logical link with capacity $\bar R_{ij} (t)$ is available at time $t$, where $\bar R_{ij} (t) = C_{ij}(\tau)$ for some finite $\tau \geq t$, where $\tau$ is the frame where the actual physical D2D transmission takes place, carrying traffic scheduled at frame $t$. Note that at frame $t$, the base station has no knowledge of $C_{ij}(\tau)$, but can still compute the average capacity $\mathbb{E}_{Z(\tau)}\lb  C_{ij}(\tau)\left| \phi_{ij}\right.\rb$ for a given link $(i,j)$, for a transmission decision. We define $\bar Z(t) = Z(\tau)$. Let $\mathcal{C}(t)$ denote the instantaneous information-theoretic capacity region of the system consisting of the channels \eqref{eq:main_channel} and the set of logical links $(i,j)$ with capacities $\bar R_{ij} (t) A_{ij}(t)$, with no knowledge of $\bar Z(t)$ at the base station\footnote{Note that the D2D link is assumed to have zero capacity if $A_{ij}(t)=0$, \emph{i.e.}, if the base station did not schedule the link $(i,j)$ at time $t$.}. A physical-layer strategy $\gamma$ is a map $\lp H(t), \Phi, \bar Z(t) \rp \mapsto \lbp R_{i} (t)\rbp_{i}$ whose output vector (interpreted as the vector of information rates delivered to users, in bits/s/Hz) satisfies $\{ R_i (t) \}_i \in \mathcal{C} (t)$ for all $i$ and $t$.

Note that even though the transmission decisions of the base station does not depend on the unknown components of the network state $Z(t)$, by allowing the rate vector $\lbp R_{i} (t)\rbp_{i}$ to be anywhere inside the instantaneous capacity region, we implicitly assume an idealized rate adaptation scenario, where once the transmission occurs, the capacity corresponding to the realization of $\bar Z(t)$ is achievable. In practice this can be implemented through incremental redundancy schemes such as hybrid ARQ.

Assume an infinite backlog of data to be transmitted to each user $i\in\mathcal{N}$. The long-term average rate of user $i$ up to time $t$ is defined as $r_i(t) = \frac{1}{t} \sum_{\tau = 1}^t R_i(\tau)$, where $R_i(\tau)$ is the rate delivered to user $i$ by the physical layer scheme $\gamma(t)$ chosen at time $t$.
The long-term throughput of user $i$ is $r_i = \liminf_{t \to \infty} r_i (t)$. Define $\mathbf{r}(t) = \{ r_i(t)\}_i$, and $\mathbf{r} = \lbp r_i \rbp_i$.

Given the stability-region-achieving D2D MAC protocol $\mathcal{I}$, and a set of physical-layer strategies $\Gamma$, at every frame $t$, the base station chooses an active set $\mathcal{A}(t)$, and a strategy $\gamma (t)\in \Gamma$ consistent with $\mathcal{A}(t)$. A scheduling policy $\pi$ is a collection of mappings 
\begin{align*}
\lp \mathbf{r}(t-1), \beta (t-1), H(t), \Phi \rp \mapsto \lp \mathcal{A}(t), \gamma(t) \rp,
\end{align*}
indexed by $t$. If $\beta^\pi$ represents the vector of arrival rates to the queues under policy $\pi$, and $\mathbf{r}^\pi$ the throughputs under policy $\pi$, then the policy $\pi$ is called \emph{stable} if $\beta^{\pi} \in \Lambda$. Our goal is to design a stable policy $\pi$ that maximizes any given concave, twice-differentiable \emph{network utility function} $U (\mathbf{r}^\pi, \beta^\pi)$ of the throughputs and the fraction of time nodes spend relaying for others\footnote{Note that since $\pi$ is stable, the relaying fraction is the same quantity as the arrival rate $\beta$.}.



\section{Downlink Physical Layer: Achievable Rates}\label{sec:phy}
In this section, we describe a class of physical layer cooperation strategies that will be used as a building block for our proposed architecture, and derive its achievable rates. We will first focus on the two-user case, where we show the approximate information-theoretic optimality of the scheme. We consider the extension to MU-MIMO in Section~\ref{sec:mumimo}.

The main idea behind the cooperation strategy is that the D2D side-channel can be used by the destination node to access a quantized version of the channel output of the relay node, which combined with its own channel output, effectively forms a MIMO system. The base station can perform signaling based on singular value decomposition over this effective MIMO channel, to form two parallel AWGN channels accessible by the destination node. Next, we describe the strategy in detail, and derive the rate it achieves.

\subsection{Cooperation Strategy} \label{sec:cooperation}
We isolate a particular user pair $(i,j)$, and without loss of generality assume $(i,j)=(1,2)$. The effective network model is given by\footnote{We focus on a particular frame $t$ to characterize the instantaneous capacity, \emph{i.e.}, the achievable rate for a given set of network parameters.}
\begin{align}
y_i &= \mathbf{h}^*_i \mathbf{x} + z_i, \; i=1,2, \;\;\;\; \bar y_1 = g_{12}x_2 + \bar z_1, \label{eq:phy_chan1}
\end{align}
where $x_2[m]$ is a function of $y_2^{m-1}$, the past receptions of user $2$, and user 1 has access to $y_1$ and $\bar y_1$.


By Wyner-Ziv Theorem \cite{WynerZiv_76}, if 
\begin{align*}
\bar R_{12} \geq \min_{\substack{p(w|y_2) \\ \what y_2(w,y_1)}: \E{\|\hat y_2 - y_2\|^2} \leq D} I(y_2;w|y_1)
\end{align*}
for a given joint distribution of channel outputs $p(y_1, y_2)$, then given a block of outputs $y_2^N$, user $1$ can recover a quantized version $\hat y_{2}^N$ of outputs such that\footnote{This is achieved by performing appropriate quantization and binning of the channel outputs at user $2$ (see \cite{WynerZiv_76} for details).} $\E{\|\hat y_2 - y_2\|^2} \leq D$.

Choosing $\mathbf{x} \sim \mathcal{CN}(\mathbf{0},\mathbf{Q})$, i.i.d. over time, we get $(y_1, y_2) \sim \mathcal{CN}\lp \mathbf{0}, \mathbf{\Sigma} \rp$ i.i.d. over time, for some covariance matrix $\Sigma = \mathbf{HQH}^*$ induced by the channel, with $\mathbf{H}=\lb \mathbf{h}_1 \; \mathbf{h}_2\rb^*$. We further choose $w = y_2+ q_{2}$, where $q_{2} \sim \mathcal{CN}(0, D)$ is independent of all other variables, and we set the mapping $\hat y_2(w, y_1) = w$. We also choose $D = \frac{\sigma^2_{2|1}}{\left| g_{12}\right|^2}$, where $\sigma^2_{2|1} = \Sigma_{22} - \Sigma_{21} \Sigma_{11}^{-1} \Sigma_{12}$ is the conditional variance of $y_2$ given $y_1$. With this set of choices, it can be shown that user $1$ can access $\hat y_2 = y_2 + q_2$, where $q_2 \sim \mathcal{CN}(0, D)$.

Once user $1$ recovers $\hat y_2$, it can construct the effective MIMO channel
\begin{align}
\mathbf{y} = \lb \begin{array}{c} y_1 \\ \hat y_2\end{array}\rb = \mathbf{Hx} + \lb \begin{array}{c} z_1 \\ z_2 + q_2\end{array}\rb. \label{eq:eff_mimo}
\end{align}
It follows that all rates $R < R_{\text{MIMO}}$ are achievable over the effective MIMO channel \eqref{eq:eff_mimo}, where 
\begin{align*}
 R_{\text{MIMO}} &= \max_{\mathrm{tr}(\mathbf{Q}) \leq 1} \log \left| \mathbf{I}_2 + \mathbf{K}^{-1}\mathbf{HQH}^*\right|,
\end{align*}
with $\mathbf{K} = \diag\lp 1, \;\; 1 + \frac{\sigma^2_{2|1}}{\left| g_{12}\right|^2}\rp$. Note that due to orthogonality of the links incoming to the destination, the encoding and decoding is significantly simplified compared to traditional Gaussian relay channel with superposition, since there is no need for complex schemes such as block Markov encoding and joint decoding, and point-to-point MIMO codes are sufficient from the point of view of the source.

Note that the MIMO channel \eqref{eq:eff_mimo} can be equivalently viewed as two parallel AWGN channels, using the singular value decomposition (SVD). It will also be useful to lower bound the rates individually achievable over these two parallel streams. Assuming $\mathbf{H}=\mathbf{USV}^*$ is an SVD, it can be shown that the rates 
\begin{align}
R_{\text{MIMO}, d} &= \log\lp 1 + \frac{s_d^2 P_d}{1 + \left| u_{2d}\right|^2\frac{\sigma^2_{2|1}}{\left| g_{12}\right|^2}}\rp, \; d=1,2 \label{eq:eff_rate1}
\end{align}
are achievable respectively\footnote{We perform the SVD on $\mathbf{H}$ directly, instead of performing on $\mathbf{K}^{-1/2}\mathbf{H}$, in order to obtain closed-form expressions for the subsequent analysis.}, over the two streams, by transmit beamforming using the matrix $\mathbf{V}$ and receive beamforming using $\mathbf{U}^*$, where $s_d$ is the $d$th singular value, $u_{2,d}$ is the $(2,d)$th element of $\mathbf{U}$, and the power allocation parameters satisfy $P_1+P_2 \leq 1$.

The next theorem shows that the gap between the rate achievable with the cooperation scheme described in the previous subsection is universally within $2$ bits/s/Hz of the capacity of the network. 
\begin{thm}\label{th:gap}
For any set of parameters $\lp \mathbf{H}, \bar R_{12}, M\rp$, the capacity $\bar C$ of the MIMO single relay channel with orthogonal links from relay to destination and from source to destination satisfies $\bar C \geq  R_{\text{MIMO}} \geq \bar C - 2$.
\end{thm}
The proof is provided in the Appendix C.

\begin{remark}
The relay channel with orthogonal links from relay to destination and from source to destination was studied by \cite{LiangVeeravalli_05} and \cite{ZahediMohseni_04}. In the former, the authors consider a relaying strategy based on decode-and-forward relaying, and focus on performance optimization problems such as optimal bandwidth allocation. The latter work focuses on linear relaying functions for such channels, and characterizes the achievable rates for scalar AWGN case. Here, we propose a relaying scheme based on compress-and-forward \cite{CoverElGamal_79} that achieves a rate that is within 2 bits/s/Hz of the information-theoretic capacity for the MIMO case. 
\end{remark}
\begin{remark}
Note that this strategy can also be implemented through quantize-map-forward relaying. Although the proposed architecture supports other relaying strategies (\emph{e.g.}, amplify-forward, decode-forward etc.), we stick with compress-forward (or quantize-map-forward implementation) due to the theoretical approximate optimality \cite{AvestimehrDiggavi_11} as well as practical feasibility, which was shown in \cite{DuarteSengupta_13} through real testbed implementation.
\end{remark}


\subsection{Cooperation with MU-MIMO} \label{sec:mumimo}
In this subsection, we demonstrate how the scheme described for two users in the previous subsection can be extended to MU-MIMO with pairs of cooperative users.


Given the set $\mathcal{N}$ of users, let us index all possible downlink streams that can be generated by the scheme by $\lp i,j,d \rp \in \mathcal{N}^2 \times \{ 1,2\}$, where $(i,j)$ is represents the cooperative pair, and $d$ represents the stream index corresponding to this pair. We assume $d \neq 2$ if $i=j$, representing the case where user $i$ is scheduled without a relay. 

By a slight abuse of notation, we assume that a schedule set $\mathcal{S} \subseteq \mathcal{N}^2 \times \{ 1,2\}$ is scheduled, consisting of such triples $\lp i,j,d\rp$, where $(i,j,d) \in \mathcal{S}$ for some $k$ if and only if $A_{ij}=1$ (note that schedule set also contains the stream index unlike active set $\mathcal{A}$). Next, consider the ``virtual users'' $\lp i,j,d \rp \in \mathcal{S}$ with the channels
\begin{single}
\begin{align*}
\tilde y_{ijd} &:= \mathbf{u}^*_{ijd} \lb \begin{array}{c} y_i \\ \hat y_j\end{array}\rb = \mathbf{u}^*_{ijd}\mathbf{H}_{ij}\mathbf{x} + \mathbf{u}^*_{ijd}\lb \begin{array}{c} z_i \\ z_j + q_{ij}\end{array}\rb := \tilde h_{ijd}^*\mathbf{x} + \tilde z_{ijd},
\end{align*}
\end{single}
\begin{double}
\begin{align*}
\tilde y_{ijd} &:= \mathbf{u}^*_{ijd} \lb \begin{array}{c} y_i \\ \hat y_j\end{array}\rb = \mathbf{u}^*_{ijd}\mathbf{H}_{ij}\mathbf{x} + \mathbf{u}^*_{ijd}\lb \begin{array}{c} z_i \\ z_j + q_{ij}\end{array}\rb \\
& := \tilde h_{ijd}^*\mathbf{x} + \tilde z_{ijd},
\end{align*}
\end{double}
where $\mathbf{H}_{ij} = \lb\mathbf{h}_i \; \mathbf{h}_j\rb^*$, and assuming $\mathbf{H}_{ij} = \mathbf{U}_{ij}\mathbf{S}_{ij}\mathbf{V}^*_{ij}$ is an SVD of $\mathbf{H}_{ij}$, $\mathbf{u}_{ijd}$ is the $k$th column of $\mathbf{U}_{ij}$. By convention, we assume that $\mathbf{U}_{ii}=\lb \begin{array}{cc} 1&0 \end{array}\rb^*$. 
The variance of $\tilde z_{ijd}$ is given by $1 + \left| u_{ijd} (2) \right|^2D_{ij}$, where $u_{ijd} (2)$ is the second element of $\mathbf{u}_{ijd}$, and $D_{ij}=\frac{\sigma^2_{j|i}}{\left| g_{ij}\right|^2}$ is the distortion introduced by quantization at node $j$. 
Note that, when $i=j$, we have $\mathbf{\tilde h}_{ijd} = \mathbf{h}_i$, and we set $\bar R_{ii} = \infty$ so that $D_{ii} = 0$.

Note that through the use of SVD over the virtual MIMO channel \eqref{eq:eff_mimo}, we have reduced the system into a set of $\left| \mathcal{S}\right|$ single-antenna virtual users with channel vectors $\frac{1}{1 + \left| u_{ijd} (2) \right|^2D_{ij}} \mathbf{\tilde h}_{ijd}$. Given such a set of channel vectors, one can implement any MU-MIMO beamforming strategy (\emph{e.g.}, zero-forcing, conjugate beamforming, SLR maximization etc.), by precoding the transmission with the corresponding beamforming matrix.

\subsection{Scaling of $\SNR$ Gain in Clustered Networks} \label{sec:scaling}
In this subsection, we consider a specific clustered network model as an example, and characterize the achievable demodulation $\SNR$ gain due to D2D cooperation for the weakest user in the network, under this model. In this analysis, we use several simplifying assumptions on the channel and network model for analytical tractability, in order to get a feel for the scale of the possible gains that can be attained through cooperation. This simplification is limited to the scope of this particular subsection, and the results in the rest of the paper do not depend on these assumptions.

Consider a network where users are clustered in a circular area of radius $r$, 
whose center is a distance $d$ away from the base station, where $r \ll d$. The users are assumed to be uniformly distributed within the circular area. In general, a network might consist of several such clusters, but here we focus on one, assuming other clusters are geographically far relative to $r$.

We assume that the downlink channel vector of user $i$ at time $t$ is modeled by\footnote{This is written for a uniform linear transmit array for simplicity, but our analysis using this model can be generalized for any array configuration.}
\begin{align*}   
\mathbf{h}_i(t) = \sqrt{\rho} \sum_{k=1}^P \xi_{i,k}(t) \mathbf{e}\lp \theta_{i,k}(t)\rp,
\end{align*}
where $\rho$ is the path loss factor (assumed constant across users in the same cluster since $r \ll d$), $P$ is the number of signal paths, $\xi_{i,k}(t) \sim \mathcal{CN}(0,1)$ is the complex path gain for the $k$th path of user $i$ at time $t$, $\theta_{i,k}$ is the angle of departure of the $k$th path of the $i$th user at time $t$, and $\mathbf{e}(\theta)$ is given by
\begin{single}
\begin{align*}
\mathbf{e}(\theta) := \lb
\begin{array}{ccccc}
1 & e^{j2\pi\Delta\cos(\theta)} & e^{j2\pi2\Delta\cos(\theta)} & \dots & e^{j2\pi(M-1)\Delta\cos(\theta)}
\end{array}
\rb^*,
\end{align*}
\end{single}
\begin{double}
\begin{align*}
\mathbf{e}(\theta) := \lb
\begin{array}{ccccc}
1 & e^{j2\pi\Delta\cos(\theta)} & \dots & e^{j2\pi(M-1)\Delta\cos(\theta)}
\end{array}
\rb^*,
\end{align*}
\end{double}
for an antenna separation $\Delta$. The path gains $\xi_{i,k} (t)$ are i.i.d. across different $i$, $k$, and $t$.

Path loss between users is modeled by $\phi_{ij} = \phi_0 d_{ij}^c$ for some constant $\phi_0$, where $d_{ij}$ is the distance between $i$ and $j$, and $c>2$ is the path loss exponent. 

For simplicity of analysis, in this example network we will assume that only one cooperative pair per time slot is scheduled. Our goal is to characterize the cooperation gains in $\SNR$ when one is allowed to choose the most suitable relay $j$ for a given destination $i$. 

Invoking \eqref{eq:eff_rate1}, we define the cooperative $\SNR$ for the pair $(i,j)$, $\SNR_{ij}^{coop}$ to be
\begin{align*}
\SNR_{ij}^{coop} := \frac{s_{ij1}^2}{1 +\left| u_{ij1} (2) \right|^2\frac{\sigma^2_{j|i}}{\left| g_{ij}\right|^2}},
\end{align*}
where $s_{ij1}$ is the first singular value corresponding to the pair $(i,j)$. Since we are interested in the achievable $\SNR$ gain, in defining this quantity, we have allocated all power to only one of the available streams, ignoring the multiplexing gain that could be achieved by scheduling two parallel streams to user $i$. The maximal non-cooperative $\SNR$ for user $i$ is given by $\SNR_i^{non-coop} := \| \mathbf{h}_i\|^2$, achieved by beamforming along the direction of $h_i$. Minimum cooperative and non-cooperative $\SNR$s in the network are respectively defined as
\begin{single}
\begin{align*}
\SNR_{\min}^{coop} &:= \min_{i \in \mathcal{N}} \SNR_{ij^*(i)}^{coop}, \;\;\;\; \SNR_{\min}^{non-coop}:= \min_{i \in \mathcal{N}} \SNR_{i}^{non-coop},
\end{align*}
\end{single}
\begin{double}
\begin{align*}
\SNR_{\min}^{coop} &:= \min_{i \in \mathcal{N}} \SNR_{ij^*(i)}^{coop}, \\
\SNR_{\min}^{non-coop}&:= \min_{i \in \mathcal{N}} \SNR_{i}^{non-coop},
\end{align*}
\end{double}
where $j^*(i) = \arg \max_{j \in \mathcal{N}} \mathbb{E}\lb \left.\SNR_{ij}^{coop} \right| \phi_{ij}, \mathbf{h}_j \rb$, which arises due to relay selection, and the expectation is taken over the D2D side-channel fading $\zeta_{ij} (t)$. 

The next theorem, whose proof is in Appendix A, summarizes our results on how the $\SNR$ of the weakest user in either case scales with the number of users $n$ in the cluster.

\begin{thm}\label{th:scaling}
\begin{single}
\begin{align*}
\lim_{n \to \infty} \Prob{\SNR_{\min}^{coop} <  \frac{1}{2}M\rho\lp \frac{1}{2}\log n - 2\log\log n\rp - 1} = O\lp e^{-\log^2 n + 2\log n} \rp,
\end{align*}
\end{single}
\begin{double}
\begin{align*}
\lim_{n \to \infty} &\Prob{\SNR_{\min}^{coop} <  \frac{1}{2}M\rho\lp \frac{1}{2}\log n - 2\log\log n\rp - 1} \\
&\qquad\qquad\qquad= O\lp e^{-\log^2 n + 2\log n} \rp,
\end{align*}
\end{double}
and
\begin{align*}
\lim_{n \to \infty} \Prob{\SNR_{\min}^{non-coop} > M\rho n^{-\frac{\gamma}{2P}}\psi(2P) } = O\lp e^{-n^{1-\gamma}} \rp,
\end{align*}
for any $0 < \gamma < 1$, where $\psi(\ell) = \lp \ell!\rp^{\frac{1}{\ell}}$, and $P$ is the number of signal paths.
\end{thm}
Theorem~\ref{th:scaling} highlights the importance of having multiple options in relay selection. In the non-cooperative case, the factor $n^{-\frac{\gamma}{2P}}$ appears due to the fact that as the number of users in the cluster grows, the minimum is taken over a larger set of users, and hence it is expected for the $\SNR$ of the weakest user to decay, in the absence of cooperation. On the other hand, in the presence of cooperation, the $\SNR$ of the weakest user actually \emph{grows}. This is due to the \emph{multiuser diversity gain}, which is present due to our ability to schedule the user with the most favorable channel conditions as a relay. In other words, as the number of users grows, so does the number of possible paths from the base station to each user, and thus the maximal $\SNR$, even when the weakest user is considered.

\section{Downlink Scheduling with Cooperation}\label{sec:scheduling}
Although our analysis of the $\SNR$ gain with relay selection in the previous section is informative of the potential gains of cooperation, one should note that its scope is limited. 
For a more thorough understanding of how to perform relay selection, we formulate the problem within the network utility maximization framework, which has been extensively studied in the context of resource allocation and scheduling problems for wireless/wired networks \cite{GeorgiadisNeely_06, LinShroff_06}. 

Note that due to interference from other D2D links as well as from external sources, not all D2D users can transmit at a given time, which implicitly imposes a constraint on relay selection. In particular, one needs to ensure that the relays can find a slot for transmission to the destination user after a finite delay, \emph{i.e.}, the relay queues remain stable. The existing cross-layer optimization algorithms, \emph{e.g.}, \cite{GeorgiadisNeely_06, LinShroff_06} (\emph{e.g.}, virtual
queues, dynamic backpressure routing etc.) are not immediately applicable to this scenario. This is firstly because our physical-layer signaling is not based on routing, and makes explicit use of the broadcast nature of the wireless medium, by using both the direct link to the destination node, and the alternate link formed by relay. Consequently, the full network cannot be abstracted into a graph with isolated links, which is widely assumed in the literature. Second, since our utility metric is a function of the average amount of relaying done by users, 
different choices of relay for the same user results in different rewards, even when the rates offered in these choices are equal. Existing formulations do not capture this generalization, which necessitates a special treatment of the downlink resource allocation problem with D2D cooperation.

To achieve this, we take an approach consisting of
\begin{enumerate}
\item A generalization of the single-user scheduling algorithm of \cite{TsibonisGeorgiadis_05} based on the maximization of the derivative of the utility function to the cooperative scenario with relay selection, MU-MIMO, and incomplete network state knowledge,
\item A relay flow control scheme integrated into scheduling, which involves explicitly imposing a set of hard linear constraints on the relaying frequency of users,
\item A novel utility metric that is specific to the cooperative architecture, exhibiting desirable fairness properties.
\end{enumerate}
In particular, the second point requires the use of a novel technique using exponential barrier functions to handle the stability constraint, and the generalizations of the first point requires several modifications to the proof of \cite{TsibonisGeorgiadis_05}.

\subsection{Utility Maximization Formulation}\label{sec:formulation}

As discussed in Section~\ref{sec:model}, our goal is to design a stable policy $\pi$ that maximizes a network utility function $U(\mathbf{r}, \beta) = \sum_{i=1}^n U_i (r_i, \beta_i)$, where $U_i: [0,\infty) \times [0, 1] \to \mathbb{R}$, for $i=1,\dots,n$, are twice continuously differentiable concave functions that are non-decreasing in the first argument, and non-increasing in the second argument. Note that unlike the existing works, the utility function is not only a function of the throughput (first argument), but also a function of the amount of relaying performed for others by the user (second argument). This definition naturally introduces a penalty each time a D2D link is scheduled, and thus the out-of-band resources are not ``free''. The utility function $U_i\lp r_i, \beta_i \rp$ then jointly captures the reward of having received an average throughput of $r_i$, and the cost of having relayed $\beta_i$ fraction of time, for user $i$. We will consider a specific form of utility function in Section~\ref{sec:utility}, and demonstrate its properties in terms of fairness and relaying cost.

Fixing the transmission strategy as the one described in Section~\ref{sec:mumimo}, the problem of selecting the pair $\lp \mathcal{A}(t), \gamma(t)\rp$ reduces to the selection of a schedule set $\mathcal{S}(t) \subseteq \mathcal{N}^2 \times \{ 1,2\}$ for every frame $t$, which specifies the active set $\mathcal{A}(t)$ as well as the stream index corresponding to each pair $(i,j)\in \mathcal{A}(t)$. The schedule set chosen by policy $\pi$ at frame $t$ will be denoted by $\mathcal{S}_\pi(t)$.

Let the network state be represented by the pair $\lp K(t), Z(t)\rp$, where $K(t) = \lp H(t), \Phi\rp$ represents the network parameters causally known at the base station, and $Z(t)$ is the fading parameter, which is unknown (all variables are as defined in Section~\ref{sec:model}, Table~\ref{tb:notation}). We assume that $K(t)$ and $Z(t)$ take values over the arbitrarily large but finite sets $\mathcal{K}$ and $\mathcal{Z}$, respectively\footnote{The finiteness assumption is made for technical convenience in proofs; however the proposed scheduling algorithm itself does not rely on this assumption. By assuming a large cardinality, one can model the general case with uncountable alphabets arbitrarily closely.}. Define
\begin{align*}
\alpha^\pi_{skz} (t) = \frac{1}{t} \sum_{\tau=1}^{t} \mathbb{I}_{\mathcal{S}_\pi (\tau) = s} \mathbb{I}_{K(\tau) = k}  \mathbb{I}_{Z(\tau) = z},
\end{align*}
for $s \subseteq \mathcal{N}^2 \times \{1,2\}$, $k \in \mathcal{K}$, and $z \in \mathcal{Z}$, and $\mathbb{I}_E$ is the indicator variable for the event $E$; \emph{i.e.}, $\alpha^\pi_{skz} (t)$ is the average fraction of time the network was in state $\lp k, z\rp$, and the policy $\pi$ chose the schedule set $s$ up to time $t$.
Under this definition, our joint scheduling/relay selection problem can be formulated as the following utility optimization problem.
\begin{align}
\text{maximize} \; \sum_{i \in \mathcal{N}} U_i \lp r_i, \beta_i\rp  \;\;\; \text{s.t.} \;\; \lp \mathbf{r}, \beta\rp \in \mathcal{R} ,\;\;\; \beta \in \Lambda, \label{opt1}
\end{align}
where $\mathcal{R}$ is such that $\lp \mathbf{r}, \beta \rp \in \mathcal{R}$ if and only if there exists a scheduling policy $\pi$ such that
\begin{single}
\begin{align*}
&\liminf_{t \to \infty} \sum_{s: i \in s_1} \sum_{k \in \mathcal{K}} \sum_{z \in \mathcal{Z}} R_{skz}^{(i)} \alpha_{skz}^\pi (t) = r_i, \;\; \limsup_{t \to \infty} \sum_{s: i \in s_2} \sum_{k \in \mathcal{K}} \sum_{z \in \mathcal{Z}} \alpha_{skz}^\pi (t) = \beta_i,
\end{align*}
\end{single}
\begin{double}
\begin{align*}
&\liminf_{t \to \infty} \sum_{s: i \in s_1} \sum_{k \in \mathcal{K}} \sum_{z \in \mathcal{Z}} R_{skz}^{(i)} \alpha_{skz}^\pi (t) = r_i, \\& \limsup_{t \to \infty} \sum_{s: i \in s_2} \sum_{k \in \mathcal{K}} \sum_{z \in \mathcal{Z}} \alpha_{skz}^\pi (t) = \beta_i,
\end{align*}
\end{double}
almost surely for all $i \in \mathcal{N}$, where $s_1 := \lbp i: (i,j,d)\in s\rbp$, $s_2 := \lbp j: (i,j,d)\in s, i \neq j\rbp$, and $R_{skz}^{(i)}$ is the rate delivered to user $i$ when $\mathcal{S}_\pi = s, \mathcal{K} = k, \mathcal{Z} = z$, which can be computed based on the results from Section~\ref{sec:phy}. Note that in the optimization problem \eqref{opt1}, the first constraint simply ensures feasibility of the pair $\lp \mathbf{r}, \beta \rp$, and the second one imposes the stability constraint for the relay queues, given the conflict graph $\mathcal{G}_c$ between the flows $(i,j)$ available in the network. 

\subsection{Stability Region Structure}
Let $\Lambda \lp \mathcal{G}_c\rp$ denote the stability region corresponding to the conflict graph $\mathcal{G}_c$. In general, an explicit characterization of $\Lambda\lp \mathcal{G}_c\rp$ is difficult to obtain. However, it turns out one can explicitly obtain a reasonably large inner bound by appropriately inserting edges in the conflict graph, and thus backing off from the optimal stability region. The following theorem characterizes this inner bound.
\begin{thm}\label{th:ssp}
Given the conflict graph $\mathcal{G}_c = \lp \mathcal{V}_c, \mathcal{E}_c\rp$ and the non-zero link availability probabilities $\lbp p_{ij}\rbp$, there exists a polynomial-time algorithm that generates another graph $\mathcal{\bar G}_c = \lp \mathcal{V}_c, \mathcal{\bar E}_c\rp$ such that $\Lambda\lp \mathcal{\bar G}_c\rp \subseteq \Lambda\lp \mathcal{G}_c\rp$, and $\beta \in \Lambda\lp \mathcal{\bar G}_c\rp$ if and only if $\beta_Q := \sum_{(i,j) \in Q} \frac{\beta_{ij}}{p_{ij}} \leq 1$ for every maximal clique\footnote{A maximal clique is a clique that is not a subset of another clique.} $Q$ of $\mathcal{\bar G}_c$. Further, the number of maximal cliques of $\mathcal{\bar G}_c$ is at most $n^2$, and these cliques can be listed in polynomial time.
\end{thm}
The proof of Theorem~\ref{th:ssp}, given in Appendix F, relies on standard results from \cite{TassiulasEphremides_92} specialized to our one-hop network consisting of user pairs, as well as certain graph-theoretic results on perfect graphs, \emph{i.e.}, graphs whose chromatic numbers equal their clique number.

The relay flow control component of our scheduling algorithm uses the inner bound of Theorem~\ref{th:ssp} to ensure the stability of the relay queues. Defining $\bar \Lambda := \Lambda\lp \mathcal{\bar G}_c\rp$, we reformulate the optimization \eqref{opt1} as
\begin{align}
\text{maximize} \; \sum_{i \in \mathcal{N}} U_i \lp r_i, \beta_i\rp \;\; \text{s.t.}  \lp \mathbf{r}, \beta\rp \in \mathcal{R},\;\;\; \beta \in \bar \Lambda.  \label{opt2} 
\end{align}
The optimality of the proposed scheduling algorithm is with respect to \eqref{opt2}.

\subsection{Optimal Scheduling}\label{sec:optimal}

Let $\mathcal{Q}$ be the set of maximal cliques of $\mathcal{G}_c$. 
Consider the following policy, which we call $\pi^*$: Given $\lp \mathbf{r}(t-1), \beta(t-1), H(t), \Phi\rp$, choose the schedule set $s^*$ such that $s^* = \argmax\limits_{s \subseteq \mathcal{\bar N}(t) \times \{ 1,2\} }  f(s)$, where
\begin{single}
\begin{align}
f(s) = \sum_{(i,j,d) \in s} &\mathbb{E}_{Z(t)} \lb \left. R_{sK(t)Z(t)}^{(i)}\right| K(t)\rb \frac{\partial U_i}{\partial r_i}\Bigr|_{\substack{r_i=r_i(t-1)\\ \beta_i=\beta_i (t-1)}}+ \frac{\partial U_j}{\partial \beta_j}\Bigr|_{\substack{r_j=r_j(t-1)\\ \beta_j=\beta_j (t-1)}}, \label{sch_rule}
\end{align}
\end{single}
\begin{double}
\begin{align}
f(s) = \sum_{(i,j,d) \in s} &\mathbb{E}_{Z(t)} \lb \left. R_{sK(t)Z(t)}^{(i)}\right| K(t)\rb \frac{\partial U_i}{\partial r_i}\Bigr|_{\substack{r_i=r_i(t-1)\\ \beta_i=\beta_i (t-1)}} \notag\\
&\qquad \qquad + \frac{\partial U_j}{\partial \beta_j}\Bigr|_{\substack{r_j=r_j(t-1)\\ \beta_j=\beta_j (t-1)}}, \label{sch_rule}
\end{align}
\end{double}
$\mathcal{\bar N}(t) :=\lbp (i,j) \in \mathcal{N}^2: \beta_{Q}(t) \leq 1 \; \text{ for all $Q \in \mathcal{Q}$ s.t. $(i,j) \in Q$}\rbp$, and $R_{sK(t)Z(t)}^{(i)} = R_{skz}^{(i)}$ with $K(t)=k$ and $Z(t)=z$. Note that $(i,i) \in  \mathcal{\bar N}(t)$ is vacuously true for all $i$, corresponding to the scenario where user $i$ is scheduled without relay.

There are a few key points to note in the definition of policy $\pi^*$. First, note that the maximization is performed over the available \emph{streams} $(i,j,d)$ in the network, as opposed to over the set of users themselves. Second, at any frame $t$, any stream $(i,j,d)$ that involves a pair of users $(i,j)$ that is part of a clique $Q$ that currently violates its constraint $\beta_Q(t) \leq 1$ is ignored in the maximization, which is the relay flow control component of the algorithm to ensure stability of the relay queues. Third, the asymptotic optimality of $\pi^*$ reveals that it is sufficient to average the rate $R_{sK(t)Z(t)}^{(i)}$ over the part of the network state $Z(t)$ that is unknown at the base station, which is consistent with the results in \cite{ShiraniCaire_10}. 

\begin{thm}\label{th:opt}
Let the optimal value of the maximization in \eqref{opt2} be $\msf{OPT}$. Define the empirical utility of $\pi^*$ as $U^* (t) = \sum_{i \in \mathcal{N}} U_i \lp r_i^* (t), \beta_i^* (t)\rp$, where $r_i^* (t)$ and $\beta_i^* (t)$ correspond to variables $r_i(t)$ and $\beta_i(t)$, respectively, under policy $\pi^*$. Then the following events hold with probability 1 (\emph{i.e.}, almost surely) in the probability space generated by the random network parameters $K(t)$ and $Z(t)$:
\begin{enumerate}
\item $\lim_{t\to \infty} \inf \lbp \| \beta^* (t) - \beta \|_1 : \beta \in \bar \Lambda \rbp = 0$,
\item $\lim_{t\to\infty} U^* (t) = \msf{OPT}$.
\end{enumerate}
\end{thm}
The proof outline is provided in Section~\ref{sec:proof}, with details in Appendix B. Theorem~\ref{th:opt} shows that policy $\pi^*$ asymptotically achieves the optimum of \eqref{opt2}. 

\subsection{Greedy Implementation}
Although converging to the optimal solution, policy $\pi^*$ suffers from high computational complexity, since it involves an exhaustive search over all subsets of streams. To reduce the complexity, we consider a suboptimal greedy implementation of the policy, similar to \cite{DimicSidiropoulos_05} for non-cooperative MU-MIMO. The algorithm works by iteratively building the schedule set, at each step adding the stream $(i^*,j^*,d^*)$ that contributes the largest amount to the objective $f(s)$, and committing to this choice in the following iterations, until there are no streams left that can result in a utility increment factor of $(1+\epsilon)$ to the existing schedule set (see Algorithm~\ref{alg:greedy}). The worst-case complexity of the algorithm is $O\lp NDn\rp$, where $D$ is the maximum node degree in $\mathcal{G}$, and $N$ is the maximum number of streams that can be scheduled at a time.
\begin{algorithm}[htbp]
\scriptsize
\begin{algorithmic}[1]
\STATE $iter=1$, $schedule\_set=\emptyset$, initialize $\epsilon>0$.
\WHILE{$iter \leq N$}
\STATE $\lp i^*, j^*, d^*\rp = \argmax_{(i,j,d) \in \mathcal{\bar N}(t) \times \{ 1,2\}} f\lp schedule\_set \cup (i,j,d)\rp$ \label{line:algmax}
\STATE $f^*(iter) = f\lp schedule\_set \cup (i^*,j^*,d^*)\rp$
\IF{$f^*(iter) > (1+\epsilon)f^*(iter-1)$}
\STATE $schedule\_set = schedule\_set \cup (i^*, j^*, d^*)$
\STATE $iter=iter+1$
\ELSE
\FORALL{$Q \in \mathcal{Q}$}
\STATE $\beta_Q(t+1) = update\_clique\_states(\beta_Q(t), schedule\_set)$
\ENDFOR
\STATE stop
\ENDIF
\ENDWHILE
\end{algorithmic}
\caption{Greedy cooperative scheduling}
\label{alg:greedy}
\end{algorithm}

\subsection{Choice of Utility Function}\label{sec:utility}
We focus on utility functions of the form\footnote{Note that this choice means that the function is not defined for $\beta_i=1$ and $r_i=0$, but we ignore this since no user will operate at these points.}
\begin{align}
U_i\lp r_i, \beta_i\rp = \log(r_i) + \kappa\log(1 - \beta_i), \label{eq:utility_form}
\end{align}
where $\kappa \geq 0$ is a parameter that controls the trade-off between fairness in throughput and fairness in relaying load. Using the concavity of the objective, it can be shown that (see Appendix~E for details) for any feasible pair $\lp \mathbf{r}, \beta \rp$, the optimum $\lp \mathbf{\wtild r}, \wtild \beta \rp$ with respect to the objective \eqref{eq:utility_form} satisfies
\begin{align}
\sum_i \frac{r_i - \wtild r_i}{\wtild r_i} \leq \kappa\sum_{i} \frac{(1-\wtild \beta_i)-(1-\beta_i)}{1-\wtild \beta_i}. \label{fairness}
\end{align}
The condition \eqref{fairness} admits a meaningful interpretation. Note that the left-hand side represents the sum of the relative gains in throughput due to the perturbation, whereas the right hand-side represents the sum of the relative decrease in time spent idle (not relaying). The condition in \eqref{fairness} then suggests that any perturbation to the optimal values will result in a total percentage throughput gain that is less than the total percentage increase in relaying cost, with the parameter $\kappa$ acting as a translation factor between throughput and relaying cost. This can be considered a generalization of well-studied proportional fairness, which implies that any perturbation to the optimal operating point results in a total percentage throughput loss. Our generalization allows for a positive total relative throughput change, albeit only at the expense of a larger total relative cost increase in relaying. For this utility function, we can evaluate the scheduling rule \eqref{sch_rule} as
\begin{align*}
s^* = \argmax_{s \subseteq \mathcal{\bar N}(t) \times \{ 1,2\} }  \frac{\mathbb{E}_{Z(t)} \lb \left. R_{sK(t)Z(t)}^{(i)}\right| K(t)\rb}{r_i (t)} - \frac{\kappa}{1 - \beta_{ij}(t)}.
\end{align*}

\subsection{Proof Outline of Theorem~\ref{th:opt}}\label{sec:proof}
We provide the outline for the proof of Theorem~\ref{th:opt}, leaving details to Appendix B.

We begin with the first claim. Due to Theorem~\ref{th:ssp}, it is sufficient to show that for any maximal clique $Q \subseteq \mathcal{V}_c$, $\limsup \beta_{Q}^* (t) \leq 1$ almost surely. We state this in the following lemma, whose proof is relatively straightforward and provided in Appendix D.

\begin{lem}\label{lem:claim1}
For all maximal cliques $Q$ of $\mathcal{G}_c$, $\limsup \beta_{Q}^* (t) \leq 1$ with probability 1 in the probability space generated by $K(t)$ and $Z(t)$.
\end{lem}

The proof of the second claim uses stochastic approximation techniques similar to the main proof in \cite{TsibonisGeorgiadis_05}, but also features several key differences to account for D2D cooperation, multiuser MIMO, partial network knowledge, relay queue stability, and generalized utility functions. To prove the second claim, we first reformulate \eqref{opt2} in terms of the variables $\alpha_{skz}$, as follows
\begin{single}
\begin{align}
\text{maximize   } & U(\mathbf{y}) := \sum_{i \in \mathcal{N}} U_i \lp  \sum_{s: i \in s_1} \sum_{k \in \mathcal{K}} \sum_{z \in \mathcal{Z}} R_{skz}^{(i)} \alpha_{skz} , \sum_{s: i \in s_2} \sum_{k \in \mathcal{K}} \sum_{z \in \mathcal{Z}} \alpha_{skz}\rp \label{opt3} \\
\text{s.t.   } & \alpha_{skz} \geq 0, \;\; \sum_{s} \alpha_{skz} \leq p_k q_z, \;\; \alpha_{skz} = q_z \sum_{z'} \alpha_{skz'}, \;\; \forall s,k,z \label{opt3cons1}\\
&\sum_{(i,j) \in Q} \sum_{s: \substack{i \in s_1\\ j \in s_2}} \sum_{k \in \mathcal{K}}  \sum_{z \in \mathcal{Z}} \alpha_{skz} \leq 1, \;\; \forall Q \in \mathcal{Q},\label{opt3cons2}
\end{align}
\end{single}
\begin{double}
\begin{align}
\text{maximize   } & U(\mathbf{y}) := \sum_{i \in \mathcal{N}} U_i \lp  \sum_{s: i \in s_1} \sum_{k \in \mathcal{K}} \sum_{z \in \mathcal{Z}} R_{skz}^{(i)} \alpha_{skz} ,\right. \notag\\
&\qquad\qquad \qquad\qquad\left. \sum_{s: i \in s_2} \sum_{k \in \mathcal{K}} \sum_{z \in \mathcal{Z}} \alpha_{skz}\rp \label{opt3} \\
\text{s.t.   } & \alpha_{skz} \geq 0, \;\; \sum_{s} \alpha_{skz} \leq p_k q_z, \;\; \alpha_{skz} = q_z \sum_{z'} \alpha_{skz'}, \notag \\
&\qquad\qquad\qquad\qquad \qquad\qquad \qquad \forall s,k,z \label{opt3cons1}\\
&\sum_{(i,j) \in Q} \sum_{s: \substack{i \in s_1\\ j \in s_2}} \sum_{k \in \mathcal{K}}  \sum_{z \in \mathcal{Z}} \alpha_{skz} \leq 1, \;\; \forall Q \in \mathcal{Q},\label{opt3cons2}
\end{align}
\end{double}
where $p_k = \Prob{ K(t) = k}$, and $q_z = \Prob{Z(t) = z}$, where $\alpha_{skz}$ are deterministic; they represent the fraction of time spent in state $(s,k,z)$ throughout the transmission. The last condition in \eqref{opt3cons1} reflects the fact that the scheduling decision cannot depend on the realization of $Z(t)$, since this information is not available at the base station.

\begin{lem}\label{lem:opt21}
Let $\mathsf{OPT}'$ denote the optimal value of \eqref{opt3}. Then $\mathsf{OPT}' \geq \mathsf{OPT}$.
\end{lem}

Lemma~\ref{lem:opt21} is proved in Appendix D using properties of compact sets.

Using Lemma~\ref{lem:opt21}, it is sufficient to show that $U^\pi (t)$ converges to the optimum value of \eqref{opt3}. We state this in the following lemma, whose proof is provided in Appendix B.
\begin{lem}\label{lem:main}
$\lim_{t \to \infty} U^* (t) = \mathsf{OPT}'$, with prob. 1 in the probability space generated by $\lp K(t), Z(t)\rp$.
\end{lem}
The proof of Lemma~\ref{lem:main} extends the stochastic approximation techniques from \cite{TsibonisGeorgiadis_05, BhattacharyaGeorgiadis_95} to our setup. In particular, we consider the relaxed version of the optimization problem by augmenting the objective with the stability constraint using a sequence of exponential barrier functions. We then determine the optimal policy for the relaxed problem, and take the limit in the slope of the barrier function to prove the result for the original problem.

\section{Simulation Results}\label{sec:sim}
\begin{table}
\centering
\begin{minipage}{.46\textwidth}
  \centering
\scriptsize
\begin{tabular}{| l | c || l | c | }
\hline
{\bf Parameter} & {\bf Value} & {\bf Parameter} & {\bf Value} \\
\hline
Cellular bandwidth & 40MHz & DL carrier freq. & 2GHz \\
D2D bandwidth &  40MHz & D2D carrier freq. & 5GHz\\
\# BS antennas & 32 (linear array)&OFDM FFT size & 2048 \\
\# UE antennas & 1 cell.+1 ISM&Power allocation & equal \\
Antenna spacing & $0.5\lambda$ &BS power & 46dBm \\
BS antenna gain & 0 dBi &UE power & 23dBm \\
BS antenna pattern & Uniform & Penetration loss & 0dB \\ 
\hline
\end{tabular} 
\caption{System parameters used in the simulations}
\label{tb:sim_params}
\end{minipage}\hfill
\begin{minipage}{.46\textwidth}
  \centering
  \scriptsize 
\begin{tabular}{|l|c|c| }
\hline
& {\bf Large Cell} & {\bf Small/Hetero.} \\
\hline
{ Inter-site distance ($a\sqrt{3}$)} & 1732m & 500m \\
{ No. cells ($\Omega$)} & 5 & 19  \\
{ No. active users/cell ($n$)} & 25 & 10  \\
{ Cluster radius std. dev. ($\sigma$)} & 20m & 10m  \\
{ Mean \# clusters} ($\frac{3\sqrt{3}}{2}\lambda a$) & 5 & 3  \\
{ Utility trade-off param. ($\kappa$)} & 7 & 8  \\
\hline
\end{tabular} 
\caption{Default cell-size-specific parameters}
\label{tb:cell_size_params}
\end{minipage}
\end{table}

\subsection{Simulation Setup}


\subsubsection{Geographic distribution}
For the regular network model, we consider a hexagonal grid of $\Omega$ cells (see Figure~\ref{fig:macro_map}), each of radius $a$, with a base station at the center, and $n$ users at each cell. For each cell, we first generate a set of cluster centers according to a homogeneous Poisson point process with intensity $\lambda$. Next, we randomly assign each user to a cluster, where user locations for cluster $i$ are chosen i.i.d. according to $\mathcal{CN}(\mathbf{c}_i, \sigma^2 \mathbf{I}_2)$, where $\mathbf{c}_i$ is the $i$'th cluster center, with $\sigma$ determining how localized the cluster is. In the heterogeneous network model (see Figure~\ref{fig:hetero_map}), we place the $\Omega$ base stations uniformly at random, generate cluster centers through a homogeneous Poisson process, and assign users to clusters uniformly at random. Next, each user associates with the nearest base station. In both cases, for each set of spatial parameters, we generate eight ``drops'', \emph{i.e.}, instantiations of user distributions, and the CDFs are computed by aggregating the results across the drops.

\subsubsection{Channel model}
For each (BS, user) pair, we generate a time series of $100$ channel vectors for each OFDM subcarrier using the 3GPP Spatial Channel Model (SCM) implementation \cite{SaloDelGado_05}, assuming a user mobility of 3m/s. For each user pair, we use the models from 3GPP D2D Channel Model \cite{TR36.843} to generate the path loss parameter $\phi_{ij}$ and the log-normal shadowing parameter $\chi_{ij}$. The channel between the user pair $(i,j)$ for each resource block (RB) is then computed as $\phi_{ij}\chi_{ij}\zeta_{ij}$, where $\zeta_{ij} \sim \mathcal{CN}(0,1)$ is i.i.d. fading parameter for a given RB. The D2D fading parameters are assumed i.i.d. across RBs. For the main results, we use the line-of-sight (LOS) model, but we also explore the effect of non-line-of-sight links later in the section. For each drop, the channels are computed and stored \emph{a priori}, and all the simulations are run for the same sequences of channel realizations.


\begin{figure}
\centering
\begin{minipage}{.46\textwidth}
  \centering
  \includegraphics[scale=0.32]{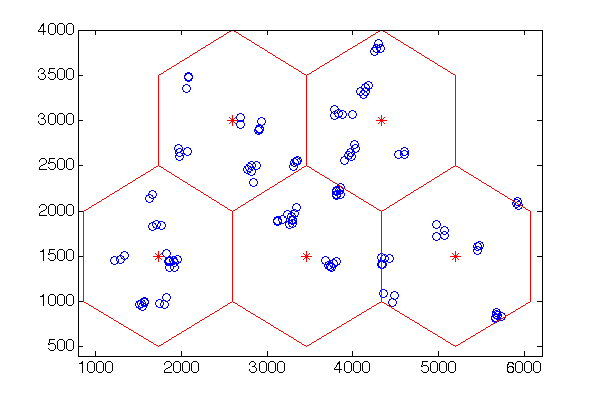}
  \caption{Sample geographic distribution of users for large cells..}
  \label{fig:macro_map}
\end{minipage}\hfill
\begin{minipage}{.46\textwidth}
  \centering
  \includegraphics[scale=0.26]{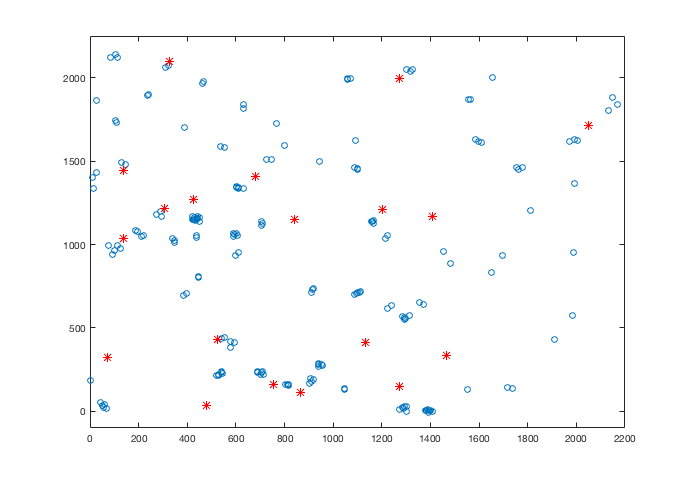}
  \caption{Example realization of usr and base station realizations for the heterogeneous network model.}
  \label{fig:hetero_map}
\end{minipage}
\end{figure}

\subsubsection{System operation}
Various system parameters are given in Table~\ref{tb:sim_params}. We assume an infinite backlog of data to be transmitted for each user. At every time slot, the base station obtains an estimate of the current network state (estimation error modeled normally distributed with variance proportional to the total energy of the channel gains across the OFDM subcarriers, independently for each antenna), and makes a scheduling decision. The scheduling decision is made without knowledge of the inter-cell interference. In the cooperative case, scheduling is done according to Algorithm 1 in Appendix E. In the non-cooperative case, we similarly use the greedy scheduling algorithm of \cite{DimicSidiropoulos_05}. Once the scheduling decision is made, the throughput is computed using the results of Section~\ref{sec:phy} based on the actual channel realizations with inter-cell interference, assuming regularized zero-forcing beamforming, and a $3$dB $\SNR$ back-off to model practical coding performance. We also take into account various rate back-offs including OFDM cyclic prefix and guard intervals, channel training and uplink data bursts. 
After the transmission, user throughputs and relaying fractions are updated through exponentially-weighted moving average filters, with averaging window $T_w = 50$ frames.

\begin{figure}
\centering
\begin{minipage}{.46\textwidth}
  \centering
  \includegraphics[scale=0.35]{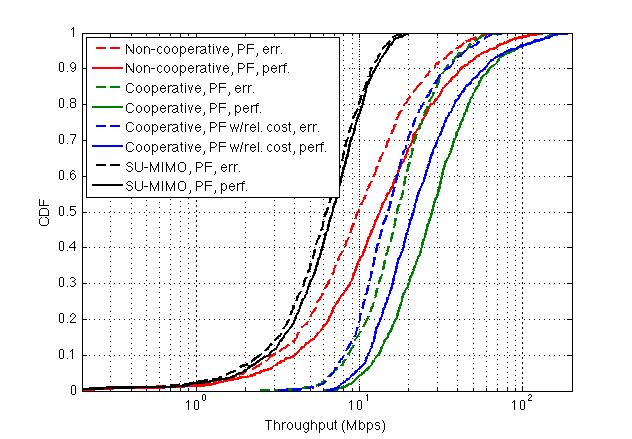}
  \caption{Throughput CDF for large cells.}
  \label{fig:macro_base}
\end{minipage}\hfill
\begin{minipage}{.46\textwidth}
  \centering
  \includegraphics[scale=0.33]{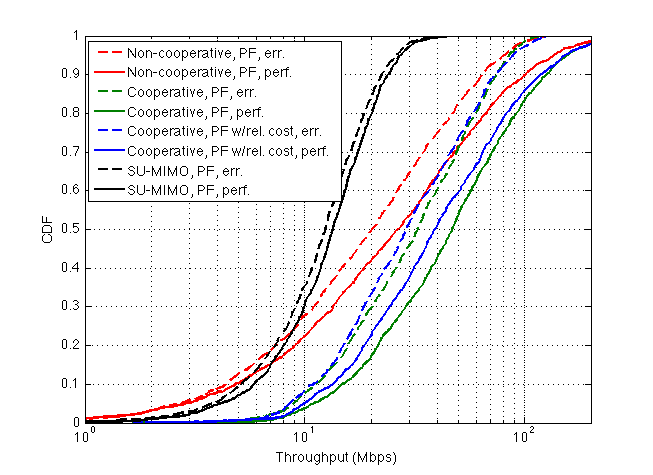}
  \caption{Throughput CDF for small cells.}
  \label{fig:micro_base}
\end{minipage}
\end{figure}

\subsection{Throughput Distribution for Regular Cells}\label{subsec:tputs}

For the setup described, we simulate the system with and without cooperation, under the utility function introduced in Section~\ref{sec:scheduling}, as well as conventional proportionally fair (PF) scheduler. We consider large and small cells, with parameters corresponding to either case provided in Table~\ref{tb:cell_size_params}. For each case, we simulate the system with and without channel estimation errors, using $p_{ij}=1$ for all $(i,j)$ (we explore smaller values of $p_{ij}$ later in the section).

The CDF of the long-term average throughput received by the users in the network is plotted in Figures~\ref{fig:macro_base} and \ref{fig:micro_base} (``err.'' represents the case with channel estimation errors, and ``perf.'' represents perfect channel estimation). These plots can be interpreted as a cumulative throughput histogram in the network, where the value on the vertical axis represents the fraction of users who experience a throughput that is less than or equal to the corresponding value on the horizontal axis.

One can observe from Figures~\ref{fig:macro_base} and \ref{fig:micro_base} that, cooperation is most helpful for the weakest (cell-edge) users in the network, providing a throughput gain ranging from $3$x up to $4.5$x for the bottom fifth-percentile of users depending on cell size, channel estimation quality and utility function used, compared to non-cooperative MU-MIMO. The gain for the median user similarly ranges from $1.4$x up to $2.1$x depending on the scenario. 

When the baseline is taken as non-cooperative SU-MIMO, the fifth percentile gain ranges from $3.5$x to $5.7$x, whereas the median gain ranges from $2.4$x up to $4.1$x.


\subsection{Throughput Distribution for Heterogeneous Networks}
We consider the same setup under the heterogeneous network model (Figure~\ref{fig:hetero_map}), with the utility function of Section~\ref{sec:scheduling}, and with the same cell-size specific parameters as those for small cells (see Table~\ref{tb:cell_size_params}). Each user associates with the closest base station, and the resulting CDF is obtained by aggregating the results from independently generated drops, where the base station locations are different across drops. We observe that similar results can be obtained for randomly placed base stations of the heterogeneous model (see Figure~\ref{fig:hetero}). The fifth-percetile gain is $4.2$x, while the median user gain is $1.8$x, with respect to non-cooperative MU-MIMO.

\begin{figure}
\centering
\begin{minipage}{.46\textwidth}
  \centering
  \includegraphics[scale=0.26]{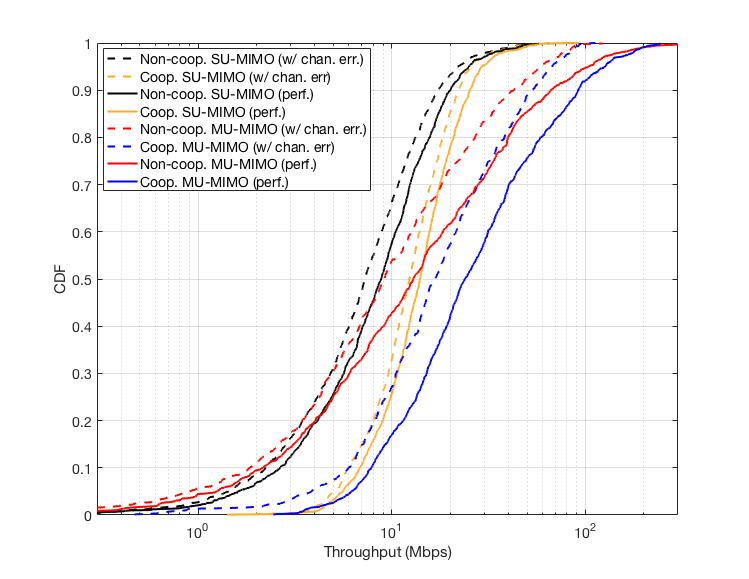}
  \caption{Throughput CDF for heterogeneous network.}
  \label{fig:hetero}
\end{minipage}\hfill
\begin{minipage}{.46\textwidth}
  \centering
    \includegraphics[scale=0.33]{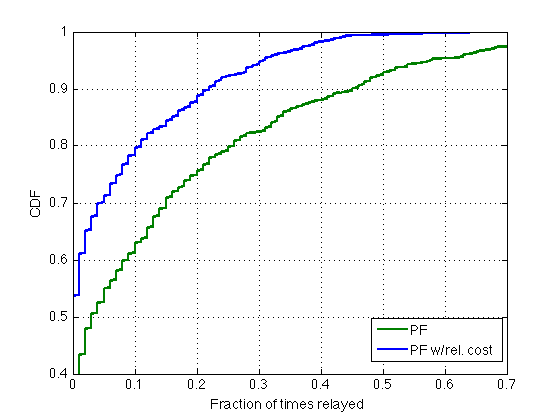}
  \caption{CDF for the fraction of time spent relaying for large cells.}
  \label{fig:macro_schd}
\end{minipage}
\end{figure}
\subsection{Relaying Cost}

We consider the CDF of the fraction of time a user has performed relaying, for the same runs of simulation as in the previous subsection, 
in Figure~\ref{fig:macro_schd}. In this figure, the values on the vertical axis represent the fraction of users that perform relaying a fraction of time less than or equal to the corresponding value at the horizontal axis, \emph{e.g.}, $90\%$ of users perform relaying less than $22\%$ of the time for PF with relaying cost, and less than $45\%$ of the time for pure PF utility. We observe that our proposed utility function results in more than $50\%$ drop in the total relaying load, with a relatively small penalty in throughput. In particular, the median throughput drop across users is 10\%, and the maximum drop is 16\%. Therefore, the novel utility function proposed in Section~\ref{sec:scheduling} enables a more efficient utilization of out-of-band resources, from a throughput-per-channel-access perspective. 

\subsection{D2D Link Intermittence}
We re-run the simulation in Subsection~\ref{subsec:tputs} for smaller values of $p_{ij}$. The results are plotted in Figure~\ref{fig:del_tol_intmt}, which suggests that the cell-edge gains are fairly robust to external interference of the D2D links, and the gains degrade gracefully with decreasing link availability, resulting in approximately $2.5$x gain at the bottom fifth percentile even when the links are only available $30\%$ of the time.

\begin{single}
\begin{figure}
\centering
\begin{minipage}{.46\textwidth}
  \centering
  \includegraphics[scale=0.33]{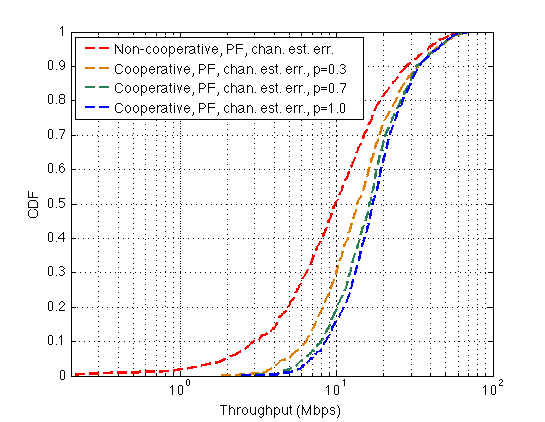}
  \caption{Throughput CDF for large cells, for intermittent side-channels.}
  \label{fig:del_tol_intmt}
\end{minipage}\hfill
\begin{minipage}{.46\textwidth}
  \centering
  \includegraphics[scale=0.28]{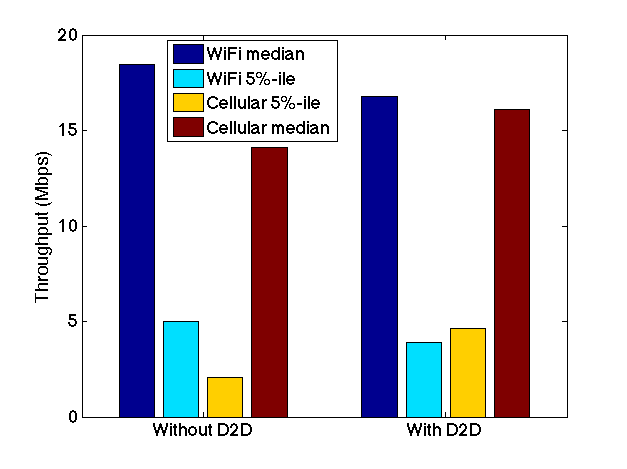}
  \caption{Throughput changes in WiFi and cellular users when D2D cooperation is enabled.}
  \label{fig:wifi_d2d}
\end{minipage}
\end{figure}
\end{single}
\begin{double}
\begin{figure}
\centering
\begin{minipage}{.46\textwidth}
  \centering
  \includegraphics[scale=0.36]{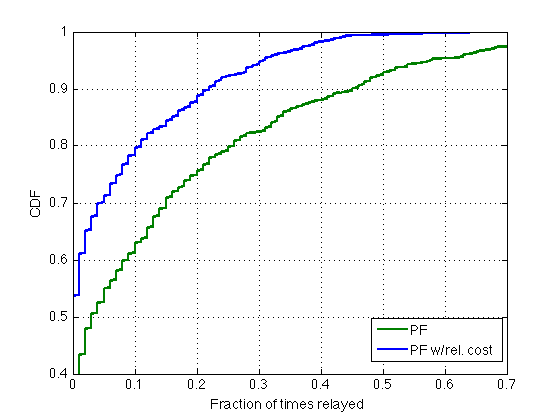}
  \caption{CDF for the fraction of time spent relaying for large cells.}
  \label{fig:macro_schd}
\end{minipage}\hfill
\begin{minipage}{.46\textwidth}
  \centering
  \includegraphics[scale=0.37]{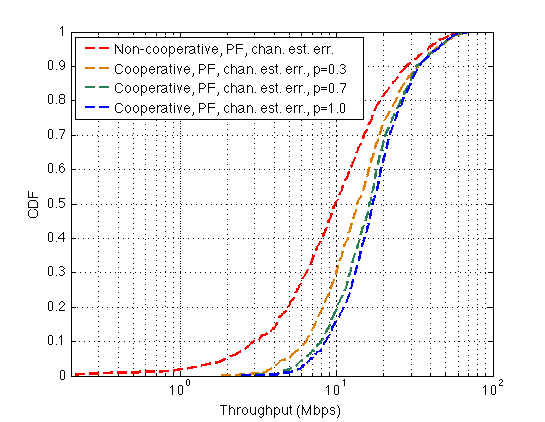}
  \caption{Throughput CDF for large cells, for intermittent side-channels.}
  \label{fig:del_tol_intmt}
\end{minipage}
\end{figure}
\end{double}

\subsection{Co-existence with WiFi}
Since the existing WiFi networks use the same band as D2D cooperation, an important question is whether co-existence of these technologies negates the possible gains due to interference. In this section, we study this scenario through simulations, and demonstrate that the combined overall benefit of WiFi access points (AP) and D2D dominates the loss due to interference, and thus WiFi and D2D cooperation can co-exist harmoniously.

To study this scenario, we consider a network model where an AP is placed at each cluster center $\mathbf{c}_i$. 
If a user is within the range of a AP, it only gets served by the AP, and is unavailable for D2D cooperation, since the unlicensed band is occupied by AP transmissions and we assume there is constant downlink traffic from the AP. Otherwise, the user is served by the base station and is potentially available for D2D cooperation. In practice, this co-existence mechanism can be implemented through a more aggressive policy, similar to LTE-U: having the user search for an available channel within the unlicensed band for a specified period of time, to use for D2D cooperation, and if none exists, having the user transmit for a short duty cycle. Note that the D2D transmissions from outside the AP range can still interfere with the receptions of AP users.

\begin{single}
\begin{figure}
\centering
\begin{minipage}{.46\textwidth}
  \centering
  \includegraphics[scale=0.3]{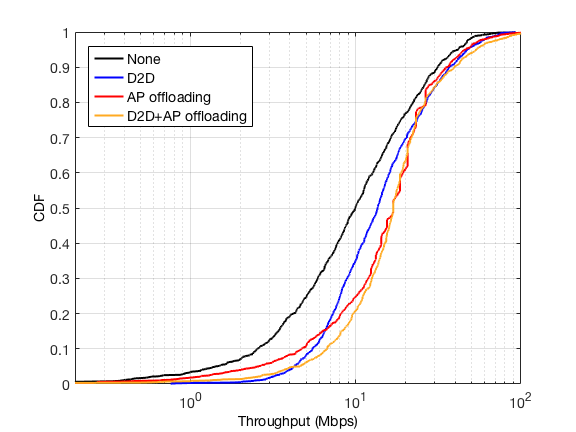}
  \caption{Throughput CDF for large cells with APs, with $\sigma=100m$.}
  \label{fig:macro_ap_100}
\end{minipage}\hfill
\begin{minipage}{.46\textwidth}
  \centering
  \includegraphics[scale=0.3]{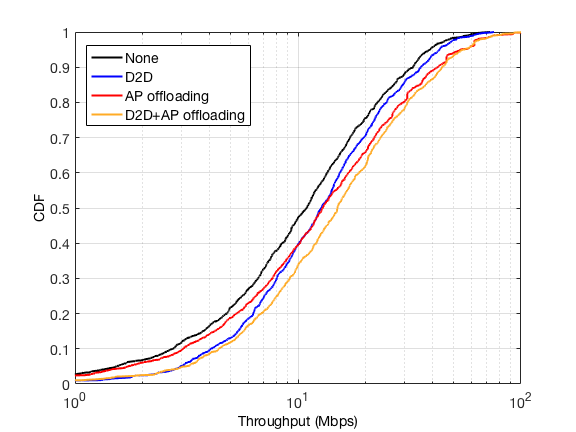}
  \caption{Throughput CDF for large cells with APs, with $\sigma=200m$.}
  \label{fig:macro_ap_200}
\end{minipage}
\end{figure}
\end{single}
\begin{double}
\begin{figure}
\centering
\begin{minipage}{.46\textwidth}
  \centering
    \includegraphics[scale=0.26]{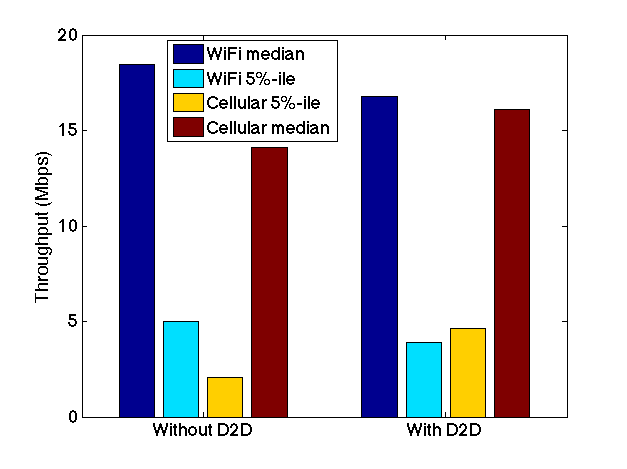}
  \caption{Throughput changes in WiFi and cellular users when D2D cooperation is enabled.}
  \label{fig:wifi_d2d}
\end{minipage}\hfill
\begin{minipage}{.46\textwidth}
  \centering
    \includegraphics[scale=0.28]{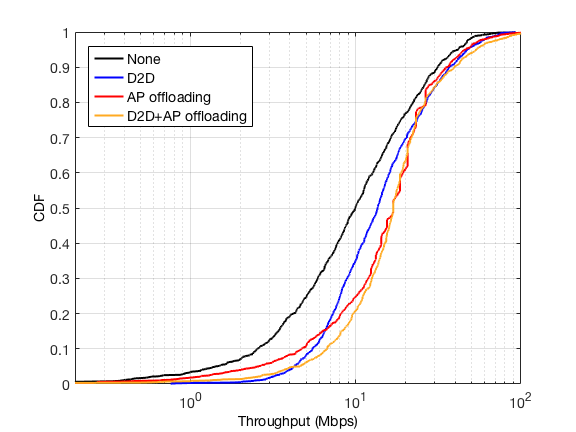}
  \caption{Throughput CDF for large cells with APs, with $\sigma=100m$.}
  \label{fig:macro_ap_100}
\end{minipage}
\end{figure}
\end{double}

We consider a simplified model for the rates delivered by the AP. If there are $\ell$ users within the range of a given AP, then a user $i$ at a distance $d_i$ from the AP is offered a rate
\begin{align*}
R_i (t) = \eta J_{i} (t) \min \lp R\lp d_i\rp, \frac{R_{\max}}{\ell}\rp,
\end{align*}
where $R\lp d\rp$ is a function that maps the user distance $d$ from AP to the rate delivered to that user, $R_{\max}$ is the maximum rate the AP can deliver, $0 < \eta \leq 1$ is a back-off factor capturing various overheads in the system, and $J_k (t)$ is the binary variable that takes the value 0 if a neighbor of $k$ in the connectivity graph is transmitting at time $t$, and 1 otherwise. We use the 802.11ac achievable rates reported in \cite{Broadcom80211ac_12} ($3$ streams, $80$MHz, with rates normalized to $40$MHz) for the $R\lp d_k\rp$ and $R_{\max}$ values, with $\eta=0.5$. We reduce the device power to $17$dBm for this setup. The throughput CDFs under this setup are given in Figures~\ref{fig:macro_ap_100} and \ref{fig:macro_ap_200}. If a user is served by WiFi, its throughput from WiFi is considered; otherwise, its throughput from the D2D-enhanced cellular network is considered.

%
%

The results suggest that when D2D cooperation and WiFi AP are simultaneously enabled, the performance is uniformly better than either of them individually enabled, despite the interference from D2D transmissions to AP users, and the relatively fewer D2D opportunities due to users being served by AP. 
Note
that this does not mean that the throughput of a given WiFi user is
not reduced when D2D interference takes place (see Figure~\ref{fig:wifi_d2d}, where median WiFi user throughput drops by 10\%, while the fifth-percentile cellular user throughput grows by $130$\%); it means that, if
the user falls within the bottom $x$-percentile after the D2D
interference, they are still better off than the bottom
$x$-percentile when only WiFi is enabled. The main reason D2D does not hurt WiFi too much is that D2D cooperation is used for a relatively small fraction of time compared to WiFi for a given user (see Figure~\ref{fig:macro_schd}, which shows 80\% of users relay less than 10\% of the time), which limits the amount of interference. This may also suggest that the more aggressive LTE-U-type policies may also be feasible.

\begin{figure}
\centering
\begin{minipage}{.46\textwidth}
  \centering
  \includegraphics[scale=0.35]{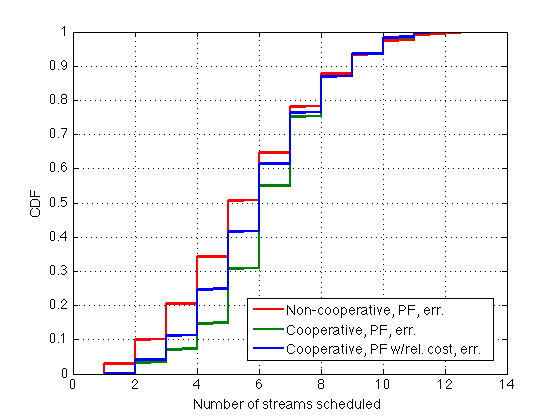}
  \caption{CDF for the number of streams scheduled for large cells.}
  \label{fig:macro_rank}
\end{minipage}\hfill
\begin{minipage}{.46\textwidth}
  \centering
  \includegraphics[scale=0.36]{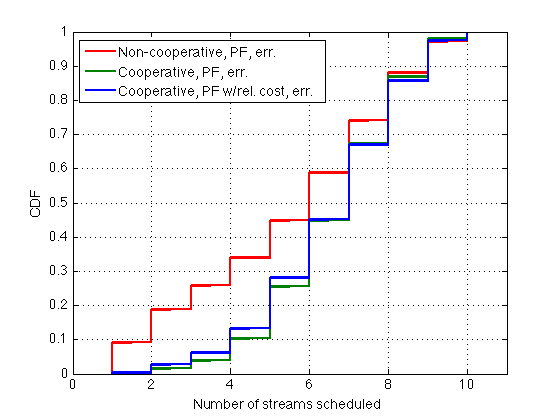}
  \caption{CDF for the number of streams scheduled for small cells.}
  \label{fig:micro_rank}
\end{minipage}
\end{figure}

\subsection{Number of Streams Scheduled}

We compare the number of streams scheduled per time slot for cooperative and non-cooperative cases, in the CDF in Figures~\ref{fig:macro_rank} and \ref{fig:micro_rank}. This can also be understood as the number of steps it takes for Algorithm 1 to terminate.

One can observe that cooperation enables the base station to schedule 1-2 additional streams on average, compared to the non-cooperative case. The reason underlying this behavior is the richness in scheduling options, since data can be transmitted to a particular user through several relaying options, with a distinct beamforming vector corresponding to each option. Since it is easier to find a stream (beamforming vector) that is compatible (approximately orthogonal) with the already scheduled streams, on the average the algorithm is able to schedule a larger number of users per time slot.

\begin{figure}
\centering
\begin{minipage}{.46\textwidth}
  \centering
  \includegraphics[scale=0.35]{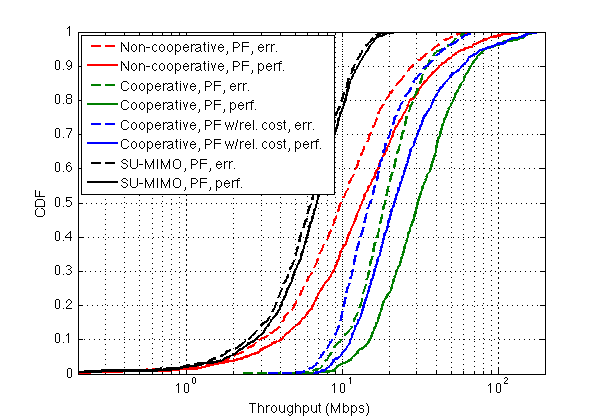}
  \caption{Throughput CDF for large cells (without stability constraint).}
  \label{fig:macro_noint}
\end{minipage}\hfill
\begin{minipage}{.46\textwidth}
  \centering
  \includegraphics[scale=0.35]{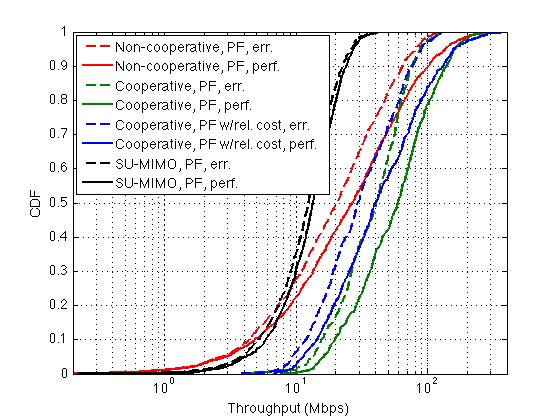}
  \caption{Throughput CDF for small cells (without stability constraint).}
  \label{fig:micro_noint}
\end{minipage}
\end{figure}

\subsection{Relaxing the Stability Constraint}

In the scenario where the cellular bandwidth is sufficiently smaller than the D2D bandwidth, the interference constraint no longer active, since the devices can perform frequency-division multiplexing to orthogonalize their transmissions. This scenario can be modeled by removing the stability constraint, and performing the maximization in \eqref{sch_rule} over all $\mathcal{N}^2\times \{ 1,2\}$ streams available for scheduling. The resulting throughput CDFs are given in Figures~\ref{fig:macro_noint} and \ref{fig:micro_noint}. Comparing the result to those in Figures~\ref{fig:macro_base} and ~\ref{fig:micro_base}, we see that the stability constraint has a rather small effect on the cooperative cell-edge gains in throughput for large cells, and a relatively larger effect for small cells. This is because the users are located more densely in small cells, and thus the interference (and thus, the stability) constraint is more restrictive. We observe that under this setup, the fifth-percentile gains with respect to SU-MIMO baseline range from $3.5\times$ up to $6.3\times$, depending on cell size, channel estimation quality and the utility function used. The median gain for large cells reaches almost $4.5\times$. The fifth-percentile gains with respect to non-cooperative MU-MIMO are similarly between $3.3\times$ and $4.9\times$, and the median user gain ranges up to $2.3\times$.

\subsection{Effect of Clustering}

For large cells, we vary the cluster radius $\sigma$ to study its effect in the throughput CDF in the network.  Figure~\ref{fig:radius} plots the throughputs corresponding to the median and the bottom fifth-percentile users in the network, for a range of cluster radii, cooperative and non-cooperative cases, and line-of-sight (LOS) and non-line-of-sight (NLOS) D2D links. We observe that at $23$dBm device power, for LOS links, most of the median and fifth-percentile throughput gains are preserved up to a cluster radius of $200$m\footnote{Note that the cluster radius is the standard deviation of user locations from each cluster center. User pairs with pairwise distance much smaller than the cluster radius can still exist within the cluster.}. The decay in throughput is much faster for NLOS D2D links, and the gain completely disappears at a cluster radius of $200$m. The performance in a real scenario would be somewhere in between the LOS and NLOS curves, since in a real scenario only a fraction of the links would be LOS.

\begin{figure}
  \centering
  \includegraphics[scale=0.35]{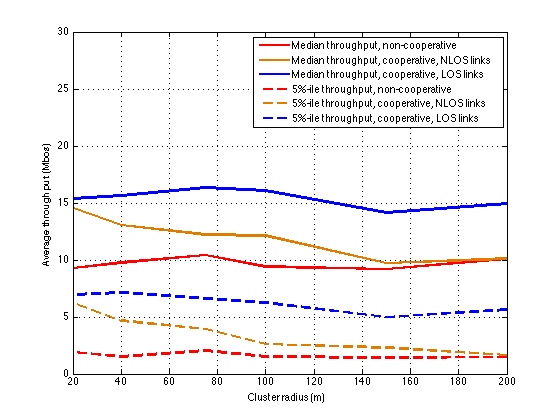}
  \caption{Median and 5-percentile throughput vs. cluster radius}
  \label{fig:radius}
\end{figure}

\subsection{Co-existence with WiFi Off-loading}
\begin{figure}
\centering
\begin{minipage}{.46\textwidth}
  \centering
  \includegraphics[scale=0.32]{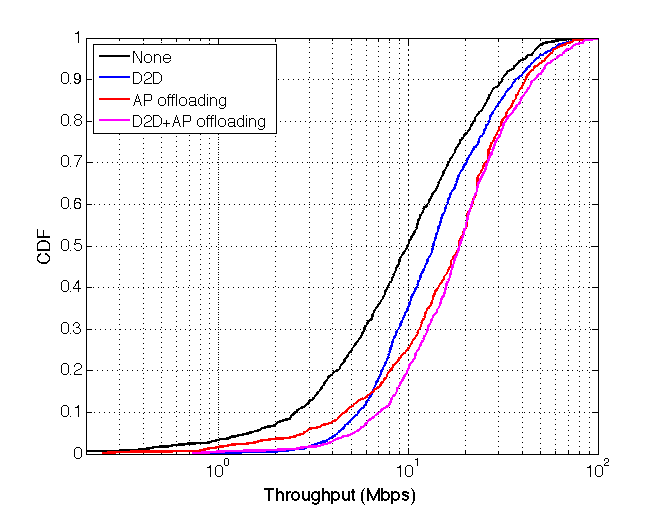}
  \caption{Throughput CDF for large cells with APs, with $\sigma=100m$ (AP users served by AP and base station).}
  \label{fig:macro_ap2_100}
\end{minipage}\hfill
\begin{minipage}{.46\textwidth}
  \centering
  \includegraphics[scale=0.29]{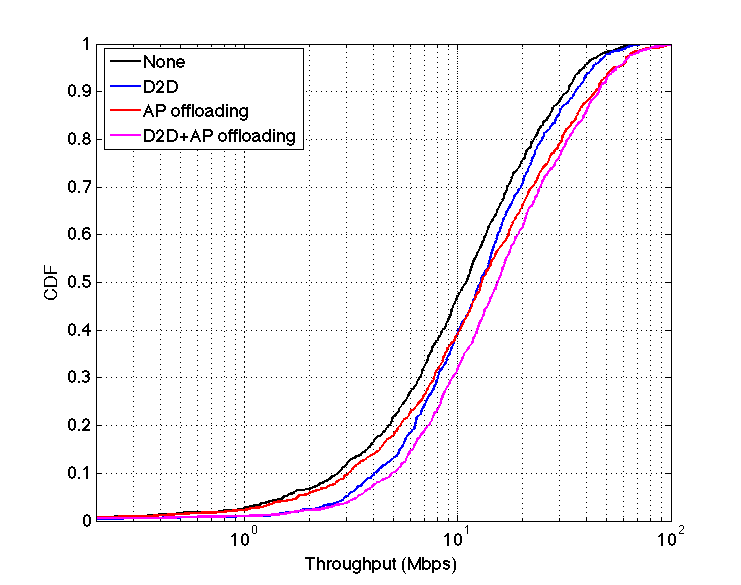}
  \caption{Throughput CDF for large cells with APs, with $\sigma=200m$ (AP users served by AP and base station).}
  \label{fig:macro_ap2_200}
\end{minipage}
\end{figure}

One can also consider an off-loading scenario where the base station continues serving the WiFi users. In this case, the WiFi users are still not available for D2D cooperation, but they can receive from both the AP and directly from the base station whenever they are scheduled based on their past throughputs. We compute the rate delivered to a WiFi user as the sum of the rate that is delivered from the base station (whenever scheduled) and the rate that is delivered from the AP. Figures~\ref{fig:macro_ap2_100} and \ref{fig:macro_ap2_200} plot the throughput CDFs under this scenario. The results follow a similar pattern to the case where WiFi users are served only by the AP, with a small additional gain in the curves with AP off-loading.

\section{Conclusion}\label{sec:conclusion}
We proposed a cellular architecture that combines MU-MIMO downlink with opportunistic use of unlicensed ISM bands to establish D2D cooperation, which  results in up to approximately $6\times$ throughput gain in cell-edge users, while improving the overall throughput. In the physical layer, the architecture is based on using D2D relaying to form virtual MIMO channels. We proposed a scheduling algorithm for this architecture that activates such D2D links to extract opportunistic gains, while maintaining fairness in terms of both throughput and the amount of relaying. To this end, we introduced a novel utility function that incorporates the cost of relaying into scheduling. We studied the architecture through extensive simulations, which suggest significant throughput gains for both cell-edge and median users under various scenarios.

\bibliography{Ref}

\begin{thebibliography}{10}

\bibitem{OttHimayat_16}
D.~Ott, N.~Himayat, and S.~Talwar, {``5G: Transforming the User Wireless
  Experience"}, {\em Towards 5G: Applications, Requirements and Candidate Technologies} pp.~34--51.
\newblock John Wiley \& Sons, Ltd, 2016.

\bibitem{ZhangWang_15}
R.~Zhang, M.~Wang, L.~X. Cai, Z.~Zheng, X.~Shen, and L.-L. Xie,
  ``{LTE}-unlicensed: the future of spectrum aggregation for cellular
  networks,'' {\em IEEE Wireless Communications}, vol.~22, no.~3, pp.~150--159,
  2015.

\bibitem{TsibonisGeorgiadis_05}
V.~Tsibonis and L.~Georgiadis, ``Optimal downlink scheduling policies for
  slotted wireless time-varying channels,'' {\em {IEEE} Transactions on
  Wireless Communications}, vol.~4, no.~4, pp.~1808--1817, 2005.


\bibitem{SendonarisErkip_03}
A.~Sendonaris, E.~Erkip, and B.~Aazhang, ``User cooperation diversity. part i.
  system description,'' {\em IEEE transactions on communications}, vol.~51,
  no.~11, pp.~1927--1938, 2003.

\bibitem{NosratiniaHunter_04}
A.~Nosratinia, T.~E. Hunter, and A.~Hedayat, ``Cooperative communication in
  wireless networks,'' {\em IEEE communications Magazine}, vol.~42, no.~10,
  pp.~74--80, 2004.

\bibitem{LiuTao_06}
P.~Liu, Z.~Tao, Z.~Lin, E.~Erkip, and S.~Panwar, ``Cooperative wireless
  communications: a cross-layer approach,'' {\em IEEE Wireless Communications},
  vol.~13, no.~4, pp.~84--92, 2006.

\bibitem{SharifHassibi_05}
M.~Sharif and B.~Hassibi, ``On the capacity of {MIMO} broadcast channels with
  partial side information,'' {\em IEEE Transactions on Information Theory},
  vol.~51, no.~2, pp.~506--522, 2005.

\bibitem{YooGoldsmith_06}
T.~Yoo and A.~Goldsmith, ``On the optimality of multiantenna broadcast
  scheduling using zero-forcing beamforming,'' {\em IEEE Journal on Selected
  Areas in Communications}, vol.~24, no.~3, pp.~528--541, 2006.


\bibitem{DimicSidiropoulos_05}
G.~Dimi{\'c} and N.~D. Sidiropoulos, ``On downlink beamforming with greedy user
  selection: performance analysis and a simple new algorithm,'' {\em IEEE
  Transactions on Signal Processing}, vol.~53, no.~10, pp.~3857--3868, 2005.

\bibitem{LiuChong_01}
X.~Liu, E.~K.~P. Chong, and N.~B. Shroff, ``Opportunistic transmission
  scheduling with resource-sharing constraints in wireless networks,'' {\em
  IEEE Journal on Selected Areas in Communications}, vol.~19, no.~10,
  pp.~2053--2064, 2001.

\bibitem{LinShroff_06}
X.~Lin, N.~B. Shroff, and R.~Srikant, ``A tutorial on cross-layer optimization
  in wireless networks,'' {\em IEEE Journal on Selected Areas in
  Communications}, vol.~24, no.~8, pp.~1452--1463, 2006.

\bibitem{GeorgiadisNeely_06}
L.~Georgiadis, M.~J. Neely, and L.~Tassiulas, {\em Resource allocation and
  cross-layer control in wireless networks}.
\newblock Now Publishers Inc., 2006.

\bibitem{ShiraniCaire_10}
H.~Shirani-Mehr, G.~Caire, and M.~J. Neely, ``{MIMO} downlink scheduling with
  non-perfect channel state knowledge,'' {\em Communications, IEEE Transactions
  on}, vol.~58, no.~7, pp.~2055--2066, 2010.

\bibitem{AsadiWang_14}
A.~Asadi, Q.~Wang, and V.~Mancuso, ``A survey on device-to-device communication
  in cellular networks,'' {\em Communications Surveys \& Tutorials, IEEE},
  vol.~16, no.~4, pp.~1801--1819, 2014.

\bibitem{DopplerRinne_09}
K.~Doppler, M.~Rinne, C.~Wijting, C.~B. Ribeiro, and K.~Hugl,
  ``Device-to-device communication as an underlay to lte-advanced networks,''
  {\em IEEE Communications Magazine}, vol.~47, no.~12, 2009.

\bibitem{LiLeiGao_12}
J.~C. Li, M.~Lei, and F.~Gao, ``Device-to-device (d2d) communication in mu-mimo
  cellular networks,'' in {\em Global Communications Conference (GLOBECOM),
  2012 IEEE}, pp.~3583--3587, IEEE, 2012.


\bibitem{WuTavildar_13}
X.~Wu, S.~Tavildar, S.~Shakkottai, T.~Richardson, J.~Li, R.~Laroia, and
  A.~Jovicic, ``{FlashLinQ}: A synchronous distributed scheduler for
  peer-to-peer ad hoc networks,'' {\em IEEE/ACM Trans. on Networking}, vol.~21, no.~4, pp.~1215--1228, 2013.


\bibitem{LiuKato_15}
J.~Liu, N.~Kato, J.~Ma, and N.~Kadowaki, ``Device-to-device communication in
  lte-advanced networks: A survey,'' {\em IEEE Communications Surveys \&
  Tutorials}, vol.~17, no.~4, pp.~1923--1940, 2015.

\bibitem{ChenZhao_14}
S.~Chen and J.~Zhao, ``The requirements, challenges, and technologies for 5g of
  terrestrial mobile telecommunication,'' {\em IEEE Communications Magazine},
  vol.~52, no.~5, pp.~36--43, 2014.

\bibitem{AsadiMancuso_13_2}
A.~Asadi and V.~Mancuso, ``On the compound impact of opportunistic scheduling
  and {D2D} communications in cellular networks,'' in {\em Proc.
  16th ACM Int. Conf. on Modeling, Analysis \& Sim. of
  Wireless and Mob. Sys.}, pp.~279--288, 2013.

\bibitem{WangRengarajan_13}
Q.~Wang and B.~Rengarajan, ``Recouping opportunistic gain in dense base station
  layouts through energy-aware user cooperation,'' in {\em IEEE 14th
  Int. Symp. on a World of Wireless, Mob. and
  Mult. Networks (WoWMoM)}, pp.~1--9, 2013.


\bibitem{Arikan_84}
E.~Arikan, ``Some complexity results about packet radio networks (corresp.),''
  {\em IEEE Transactions on Information Theory}, vol.~30, no.~4, pp.~681--685,
  1984.

\bibitem{TassiulasEphremides_92}
L.~Tassiulas and A.~Ephremides, ``Stability properties of constrained queueing
  systems and scheduling policies for maximum throughput in multihop radio
  networks,'' {\em IEEE Transactions on Automatic Control}, vol.~37, no.~12,
  pp.~1936--1948, 1992.

\bibitem{LibinWalrand_10}
L.~Jiang and J.~Walrand, ``A distributed {CSMA} algorithm for throughput and
  utility maximization in wireless networks,'' {\em IEEE/ACM Transactions on
  Networking}, vol.~18, no.~3, pp.~960--972, 2010.

\bibitem{ModianoShah_06}
E.~Modiano, D.~Shah, and G.~Zussman, ``Maximizing throughput in wireless
  networks via gossiping,'' in {\em ACM SIGMETRICS Performance Evaluation
  Review}, vol.~34, pp.~27--38, ACM, 2006.

\bibitem{WynerZiv_76}
A.~D. Wyner and J.~Ziv, ``The rate-distortion function for source coding with
  side information at the decoder,'' {\em IEEE Transactions on Information
  Theory}, vol.~22, no.~1, pp.~1--10, 1976.

\bibitem{LiangVeeravalli_05}
Y.~Liang and V.~V. Veeravalli, ``Gaussian orthogonal relay channels: Optimal
  resource allocation and capacity,'' {\em IEEE Transactions on Information
  Theory}, vol.~51, no.~9, pp.~3284--3289, 2005.

\bibitem{ZahediMohseni_04}
S.~Zahedi, M.~Mohseni, and A.~El~Gamal, ``On the capacity of {AWGN} relay
  channels with linear relaying functions,'' in {\em IEEE International
  Symposium on Information Theory (ISIT)}, pp.~399--399, 2004.

\bibitem{CoverElGamal_79}
T.~M. Cover and A.~E. Gamal, ``Capacity theorems for the relay channel,'' {\em
  IEEE Transactions on Information Theory}, vol.~25, no.~5, pp.~572--584, 1979.

\bibitem{AvestimehrDiggavi_11}
A.~S. Avestimehr, S.~N. Diggavi, and D.~N. Tse, ``Wireless network information
  flow: A deterministic approach,'' {\em Information Theory, IEEE Transactions
  on}, vol.~57, no.~4, pp.~1872--1905, 2011.

\bibitem{DuarteSengupta_13}
M.~Duarte, A.~Sengupta, S.~Brahma, C.~Fragouli, and S.~Diggavi,
  ``Quantize-map-forward (qmf) relaying: an experimental study,'' in {\em
  Proc. 14th ACM Int. Symposium on Mobile ad hoc
  networking and computing}, pp.~227--236, ACM, 2013.


\bibitem{BhattacharyaGeorgiadis_95}
P.~P. Bhattacharya, L.~Georgiadis, and P.~Tsoucas, ``Problems of adaptive
  optimization in multiclass {M/GI/1} queues with bernoulli feedback,'' {\em
  Mathematics of Operations Research}, vol.~20, no.~2, pp.~355--380, 1995.

\bibitem{SaloDelGado_05}
J.~Salo, G.~Del~Galdo, J.~Salmi, P.~Kyösti, M.~Milojevic, D.~Laselva, and
  C.~Schneider, ``{MATLAB} implementation of the {3GPP} spatial channel model
  ({3GPP TR 25.996}),'' Jan. 2005.

\bibitem{TR36.843}
``Study on {LTE} device to device proximity services; radio aspects,'' Tech.
  Rep. TR 36.843, {3GPP}, Mar 2014.

\bibitem{Broadcom80211ac_12}
``World’s first {5G WiFi 802.11ac SoC},'' tech. rep., {Broadcom Corporation},
  2012.

  \bibitem{Rebennack_08}
S.~Rebennack, ``Stable set problem: Branch \& cut algorithms,'' in {\em Encyclopedia of Optimization},
  pp.~3676--3688, Springer, 2008.

\bibitem{Chtaval_75}
V.~Chv{\'a}tal, ``On certain polytopes associated with graphs,'' {\em Journal
  of Combinatorial Theory, Series B}, vol.~18, no.~2, pp.~138--154, 1975.

\bibitem{RosgenStewart_07}
B.~Rosgen and L.~Stewart, ``Complexity results on graphs with few cliques,''
  {\em Discrete Mathematics and Theoretical Computer Science}, vol.~9, no.~1,
  2007.

\bibitem{Gavril_74}
F.~Gavril, ``The intersection graphs of subtrees in trees are exactly the
  chordal graphs,'' {\em Journal of Combinatorial Theory, Series B}, vol.~16,
  no.~1, pp.~47--56, 1974.

\bibitem{Berge_61}
C.~Berge, ``F{\"a}rbung von graphen, deren s{\"a}mtliche bzw. deren ungerade
  kreise starr sind,'' {\em Wiss. Z. Martin-Luther-Univ. Halle-Wittenberg
  Math.-Natur. Reihe}, vol.~10, no.~114, p.~88, 1961.


\end{thebibliography}


\begin{thebibliography}{10}

\bibitem{TsibonisGeorgiadis_05}
V.~Tsibonis and L.~Georgiadis, ``Optimal downlink scheduling policies for
  slotted wireless time-varying channels,'' {\em {IEEE} Transactions on
  Wireless Communications}, vol.~4, no.~4, pp.~1808--1817, 2005.

\bibitem{BhattacharyaGeorgiadis_95}
P.~P. Bhattacharya, L.~Georgiadis, and P.~Tsoucas, ``Problems of adaptive
  optimization in multiclass {M/GI/1} queues with bernoulli feedback,'' {\em
  Mathematics of Operations Research}, vol.~20, no.~2, pp.~355--380, 1995.

\bibitem{DimicSidiropoulos_05}
G.~Dimi{\'c} and N.~D. Sidiropoulos, ``On downlink beamforming with greedy user
  selection: performance analysis and a simple new algorithm,'' {\em IEEE
  Transactions on Signal Processing}, vol.~53, no.~10, pp.~3857--3868, 2005.

\bibitem{KarakusDiggavi_16}
C.~Karakus and S.~Diggavi, ``Device-to-device side-channel cooperation for
  {MIMO} downlink,'' {\em Arxiv Preprint}, October 2016,
  \url{http://arxiv.org}.

\bibitem{AsadiWang_14}
A.~Asadi, Q.~Wang, and V.~Mancuso, ``A survey on device-to-device communication
  in cellular networks,'' {\em Communications Surveys \& Tutorials, IEEE},
  vol.~16, no.~4, pp.~1801--1819, 2014.

\bibitem{AsadiMancuso_13_2}
A.~Asadi and V.~Mancuso, ``On the compound impact of opportunistic scheduling
  and {D2D} communications in cellular networks,'' in {\em Proceedings of the
  16th ACM international conference on Modeling, analysis \& simulation of
  wireless and mobile systems}, pp.~279--288, ACM, 2013.

\bibitem{YooGoldsmith_06}
T.~Yoo and A.~Goldsmith, ``On the optimality of multiantenna broadcast
  scheduling using zero-forcing beamforming,'' {\em IEEE Journal on Selected
  Areas in Communications}, vol.~24, no.~3, pp.~528--541, 2006.

\bibitem{SharifHassibi_05}
M.~Sharif and B.~Hassibi, ``On the capacity of {MIMO} broadcast channels with
  partial side information,'' {\em IEEE Transactions on Information Theory},
  vol.~51, no.~2, pp.~506--522, 2005.

\bibitem{ShiraniCaire_10}
H.~Shirani-Mehr, G.~Caire, and M.~J. Neely, ``{MIMO} downlink scheduling with
  non-perfect channel state knowledge,'' {\em Communications, IEEE Transactions
  on}, vol.~58, no.~7, pp.~2055--2066, 2010.

\bibitem{GeorgiadisNeely_06}
L.~Georgiadis, M.~J. Neely, and L.~Tassiulas, {\em Resource allocation and
  cross-layer control in wireless networks}.
\newblock Now Publishers Inc., 2006.

\bibitem{EryilmazOzdaglar_07}
A.~Eryilmaz, A.~Ozdaglar, and E.~Modiano, ``Polynomial complexity algorithms
  for full utilization of multi-hop wireless networks,'' in {\em 26th IEEE
  International Conference on Computer Communications (INFOCOM)}, pp.~499--507,
  IEEE, 2007.

\bibitem{Arikan_84}
E.~Arikan, ``Some complexity results about packet radio networks (corresp.),''
  {\em IEEE Transactions on Information Theory}, vol.~30, no.~4, pp.~681--685,
  1984.

\bibitem{TassiulasEphremides_92}
L.~Tassiulas and A.~Ephremides, ``Stability properties of constrained queueing
  systems and scheduling policies for maximum throughput in multihop radio
  networks,'' {\em IEEE Transactions on Automatic Control}, vol.~37, no.~12,
  pp.~1936--1948, 1992.

\bibitem{LibinWalrand_10}
L.~Jiang and J.~Walrand, ``A distributed {CSMA} algorithm for throughput and
  utility maximization in wireless networks,'' {\em IEEE/ACM Transactions on
  Networking (ToN)}, vol.~18, no.~3, pp.~960--972, 2010.

\bibitem{ModianoShah_06}
E.~Modiano, D.~Shah, and G.~Zussman, ``Maximizing throughput in wireless
  networks via gossiping,'' in {\em ACM SIGMETRICS Performance Evaluation
  Review}, vol.~34, pp.~27--38, ACM, 2006.

\bibitem{HolmOien_03}
H.~Holm, G.~E. Oien, M.-S. Alouini, D.~Gesbert, and K.~J. Hole, ``Optimal
  design of adaptive coded modulation schemes for maximum average spectral
  efficiency,'' in {\em 4th IEEE Workshop on Signal Processing Advances in
  Wireless Communications (SPAWC), 2003}, pp.~403--407, IEEE, 2003.

\bibitem{WynerZiv_76}
A.~D. Wyner and J.~Ziv, ``The rate-distortion function for source coding with
  side information at the decoder,'' {\em IEEE Transactions on Information
  Theory}, vol.~22, no.~1, pp.~1--10, 1976.

\bibitem{LiangVeeravalli_05}
Y.~Liang and V.~V. Veeravalli, ``Gaussian orthogonal relay channels: Optimal
  resource allocation and capacity,'' {\em IEEE Transactions on Information
  Theory}, vol.~51, no.~9, pp.~3284--3289, 2005.

\bibitem{ZahediMohseni_04}
S.~Zahedi, M.~Mohseni, and A.~El~Gamal, ``On the capacity of {AWGN} relay
  channels with linear relaying functions,'' in {\em IEEE International
  Symposium on Information Theory (ISIT)}, pp.~399--399, 2004.

\bibitem{CoverElGamal_79}
T.~M. Cover and A.~E. Gamal, ``Capacity theorems for the relay channel,'' {\em
  IEEE Transactions on Information Theory}, vol.~25, no.~5, pp.~572--584, 1979.

\bibitem{AvestimehrDiggavi_11}
A.~S. Avestimehr, S.~N. Diggavi, and D.~N. Tse, ``Wireless network information
  flow: A deterministic approach,'' {\em Information Theory, IEEE Transactions
  on}, vol.~57, no.~4, pp.~1872--1905, 2011.

\bibitem{ChaliseVandendorpe_09}
B.~K. Chalise and L.~Vandendorpe, ``{MIMO} relay design for
  multipoint-to-multipoint communications with imperfect channel state
  information,'' {\em IEEE Transactions on Signal Processing}, vol.~57, no.~7,
  pp.~2785--2796, 2009.

\bibitem{GuanLuo_08}
W.~Guan and H.~Luo, ``Joint {MMSE} transceiver design in non-regenerative mimo
  relay systems,'' {\em IEEE Communications Letters}, vol.~12, no.~7,
  pp.~517--519, 2008.

\bibitem{DuarteSengupta_13}
M.~Duarte, A.~Sengupta, S.~Brahma, C.~Fragouli, and S.~Diggavi,
  ``Quantize-map-forward (qmf) relaying: an experimental study,'' in {\em
  Proceedings of the fourteenth ACM international symposium on Mobile ad hoc
  networking and computing}, pp.~227--236, ACM, 2013.

\bibitem{LinShroff_06}
X.~Lin, N.~B. Shroff, and R.~Srikant, ``A tutorial on cross-layer optimization
  in wireless networks,'' {\em IEEE Journal on Selected Areas in
  Communications}, vol.~24, no.~8, pp.~1452--1463, 2006.

\bibitem{Kelly_97}
F.~Kelly, ``Charging and rate control for elastic traffic,'' {\em European
  Transactions on Telecommunications}, vol.~8, no.~1, pp.~33--37, 1997.

\bibitem{Rebennack_08}
S.~Rebennack, ``Stable set problem: Branch \& cut algorithms stable set
  problem: Branch \& cut algorithms,'' in {\em Encyclopedia of Optimization},
  pp.~3676--3688, Springer, 2008.

\bibitem{Chtaval_75}
V.~Chv{\'a}tal, ``On certain polytopes associated with graphs,'' {\em Journal
  of Combinatorial Theory, Series B}, vol.~18, no.~2, pp.~138--154, 1975.

\bibitem{RosgenStewart_07}
B.~Rosgen and L.~Stewart, ``Complexity results on graphs with few cliques,''
  {\em Discrete Mathematics and Theoretical Computer Science}, vol.~9, no.~1,
  2007.

\bibitem{Gavril_74}
F.~Gavril, ``The intersection graphs of subtrees in trees are exactly the
  chordal graphs,'' {\em Journal of Combinatorial Theory, Series B}, vol.~16,
  no.~1, pp.~47--56, 1974.

\bibitem{SaloDelGado_05}
J.~Salo, G.~Del~Galdo, J.~Salmi, P.~Kyösti, M.~Milojevic, D.~Laselva, and
  C.~Schneider, ``{MATLAB} implementation of the {3GPP} spatial channel model
  ({3GPP TR 25.996}),'' Jan. 2005.

\bibitem{TR36.843}
``Study on {LTE} device to device proximity services; radio aspects,'' Tech.
  Rep. TR 36.843, {3GPP}, Mar 2014.

\bibitem{Broadcom80211ac_12}
``World’s first {5G WiFi 802.11ac SoC},'' tech. rep., {Broadcom Corporation},
  2012.

\bibitem{ZhangWang_15}
R.~Zhang, M.~Wang, L.~X. Cai, Z.~Zheng, X.~Shen, and L.-L. Xie,
  ``{LTE}-unlicensed: the future of spectrum aggregation for cellular
  networks,'' {\em IEEE Wireless Communications}, vol.~22, no.~3, pp.~150--159,
  2015.

\end{thebibliography}

\begin{appendices}
\section{Proof of Theorem~\ref{th:scaling}}\label{ap:scaling}
\begin{prop}\label{prop:min_scaling}
Let $X_i, i=1,\dots,n$, be i.i.d. $\chi^2(2P)$ random variables. Then
\begin{align*}
\Prob{ \min_{1\leq i \leq n} X_i > n^{-\frac{\gamma}{2P}}\psi(2P) } = O\lp e^{-n^{1-\gamma}}\rp,\;\text{ for } 0<\gamma<1.
\end{align*}
\end{prop}
\begin{proof}
Using the Taylor series for the upper incomplete Gamma function, as $x \to 0$,
\begin{align*}
\Prob{X_i > x} = 1 - \frac{x^{2P}}{(2P)!} + O\lp x^{2P+1}\rp.
\end{align*}
Therefore,
\begin{align*}
\Prob{ \min_{1\leq i \leq n} X_i > n^{-\frac{\gamma}{2P}}\psi(2P) } = \lp \Prob{ X_i > n^{-\frac{\gamma}{2P}}\psi(2P) }\rp^n = \lp 1 - n^{-\gamma}\rp^n = O\lp e^{-n^{1-\gamma}}\rp.
\end{align*}
\end{proof}

We will first derive a lower bound on $\SNR_{ij}^{coop}$, defined by $
\SNR_{ij}^{coop} = \frac{s_{ij1}^2}{1 +\left| u_{ij1} (2) \right|^2\frac{\sigma^2_{j|i}}{\| g_{ij}\|^2}}$.
Using the fact that $\left| u_{ij1} (2) \right|^2 \leq 1$ and $\sigma^2_{j|i} \leq \sigma^2_j$, where $\sigma^2_j$ is the variance of $y_2$,
\begin{align}
\SNR_{ij}^{coop} &\geq \frac{s_{ij1}^2 }{1 + \frac{\sigma^2_j}{\| g_{ij}\|^2}} = \frac{s_{ij1}^2 }{1 + \frac{1 + \| \mathbf{h}_j\|^2}{\| g_{ij}\|^2}} \label{eq:snr_coop_lb}.
\end{align}
Next, since $s_{ij1}^2$ is the larger eigenvalue of the matrix $\mathbf{H}_{ij}\mathbf{H}_{ij}^*$, using the closed form expressions for the eigenvalues of $2 \times 2$ matrices,
\begin{align*}
s_{ij1}^2 &= \frac{1}{2}\lp \|\mathbf{h}_i\|^2 + \|\mathbf{h}_j\|^2 + \sqrt{\|\mathbf{h}_i\|^4 + \|\mathbf{h}_j\|^4 + 2\|\mathbf{h}_i\|^2\|\mathbf{h}_j\|^2\cos(2\Theta)} \rp \\
&\geq \frac{1}{2}\lp \|\mathbf{h}_i\|^2 + \|\mathbf{h}_j\|^2 + \left| \|\mathbf{h}_i\|^2 - \|\mathbf{h}_j\|^2 \right| \rp= \max\lp \|\mathbf{h}_i\|^2, \|\mathbf{h}_j\|^2\rp,
\end{align*}
where $\Theta = \cos^{-1}\frac{\mathbf{h}_i^*\mathbf{h}_j}{\| \mathbf{h}_i\|\| \mathbf{h}_j\|}$ is the angle between $\mathbf{h}_i$ and $\mathbf{h}_j$, and the lower bound is obtained by setting $\cos\lp 2\Theta\rp = -1$. Using this lower bound in \eqref{eq:snr_coop_lb}, we get
\begin{align*}
\SNR_{ij}^{coop} &\geq \frac{ \max\lp \|\mathbf{h}_i\|^2, \|\mathbf{h}_j\|^2 \rp}{1 + \frac{1 + \| \mathbf{h}_j\|^2}{\| g_{ij}\|^2}} \geq \frac{  \|\mathbf{h}_j\|^2 }{1 + \frac{1 + \| \mathbf{h}_j\|^2}{\| g_{ij}\|^2}} \geq \frac{\lp \|\mathbf{h}_j\|^2+1 \rp \| g_{ij}\|^2}{ 1 + \|\mathbf{h}_j\|^2 + \| g_{ij}\|^2} -1 \\
&\geq \frac{1}{2} \min\lp \|\mathbf{h}_j\|^2, \| g_{ij}\|^2\rp - 1.
\end{align*}
Therefore, to prove the first claim in Theorem~\ref{th:scaling}, it is sufficient to prove that
\begin{align*}
\Prob{  \min_{i \in \mathcal{N}} \min\lp \|\mathbf{h}_{j^*(i)}\|^2, \| g_{ij^*(i)}\|^2\rp > M\rho \lp  \frac{1}{2}\log n - 2\log\log n\rp } = O\lp e^{-\log^2 n + 2\log n}\rp.
\end{align*}
Define $\mathcal{P}_n = \{ j: \| \mathbf{h}_j\|^2 \geq M\rho \lp  \frac{1}{2}\log n - 2\log\log n\rp\}$, and $\mathcal{R}_n(i) = \{j \in \mathcal{P}_n: \phi_{ij} \geq n^{\frac{c}{4}}  \}$. 
\begin{prop}\label{prop:Rn_nonempty}
$\Prob{\mathcal{R}_n (i) = \varnothing \text{ for some $i$}} = O\lp e^{-\log^2 n + 2\log n}\rp$.
\end{prop}
Therefore, if $\mathcal{R}_n (i) \neq \varnothing$ for all $i$,
\begin{align*}
1 + \SNR_{\min}^{coop} &\geq \frac{1}{2} \min_{i \in \mathcal{N}} \min\lp \|\mathbf{h}_{j^*(i)}\|^2, \| g_{ij^*(i)}\|^2\rp \\
& \geq \frac{1}{2} \min_{i \in \mathcal{N}} \min\lp M\rho \lp  \frac{1}{2}\log n - 2\log\log n\rp, n^{\frac{c}{4}}\| \zeta_{ij^\dagger(i)}\|^2 \rp\\
&=  \frac{1}{2} \min\lp M\rho \lp  \frac{1}{2}\log n - 2\log\log n\rp, n^{\frac{c}{4}}  \min_{i \in \mathcal{N}}\| \zeta_{ij^\dagger(i)}\|^2 \rp,
\end{align*}
where $j^\dagger(i) = \arg \max_{j \in \mathcal{R}_n(i)} \mathbb{E}\lb \left.\SNR_{ij}^{coop} \right| \phi_{ij}, \mathbf{h}_j \rb$, and thus
\begin{align}
\Prob{ \left. \SNR_{\min}^{coop} < \frac{1}{2}M\rho \lp  \frac{1}{2}\log n - 2\log\log n\rp-1 \right| \mathcal{R}_n (i) \neq \varnothing \; \forall i} = O\lp e^{-n^{1-\gamma}}\rp, \label{eq:gamma_bound}
\end{align}
for all $0<\gamma<1$, by Proposition~\ref{prop:min_scaling}, by the fact that $\|\zeta_{ij}\|^2$ is a $\chi^2(2)$ random variable, and that $j^\dagger(i)$ is independent of $\|\zeta_{ij}\|^2$. Then \eqref{eq:gamma_bound}, together with Proposition~\ref{prop:Rn_nonempty} implies the first claim of the theorem.

It remains to prove Proposition~\ref{prop:Rn_nonempty}. To achieve this, we will first lower bound the tail probability $\Prob{\| \mathbf{h}_j\|^2 > a}$. Define $\mathbf{\hat e}_{j,k} :=\frac{\mathbf{e}_{j,k}}{\| \mathbf{e}_{j,k}\|} = \frac{\mathbf{e}_{j,k}}{\sqrt{M}}$, $\mathbf{E}_j := \lb \mathbf{\hat e}_{j,1}\; \dots \; \mathbf{\hat e}_{j,P}\rb$, and $\xi_j := \lb \xi_{j,k} \rb_{k}$. Letting $\mathbf{E}_j = \mathbf{Q}_j\Lambda_j\mathbf{Q}_j^*$ be an eigendecomposition of $\mathbf{E}_j$,
\begin{align*}
\| \mathbf{h}_j\|^2 &= \rho\left\| \sum_{k=1}^P \xi_{j,k} \mathbf{e}(\theta_{j,k})\right\|^2 = M\rho \left\| \sum_{k=1}^P \xi_{j,k} \mathbf{\hat e}(\theta_{j,k})\right\|^2 \\
&= M\rho \lp \mathbf{E}_j \xi_j\rp^*\lp \mathbf{E}_j \xi_j\rp = M\rho \xi_j^* \lp \mathbf{E}_j^*\mathbf{E}_j\rp \xi_j = M\rho \sum_{k=1}^P \lambda_k\lp \mathbf{E}_j^*\mathbf{E}_j\rp \left| \lp \mathbf{Q}_j\xi_j\rp_k\right|^2,
\end{align*}
where $\lambda_k\lp \mathbf{E}_j^*\mathbf{E}_j\rp$ is the $k$th eigenvalue of $\mathbf{E}_j^*\mathbf{E}_j$, and $\lp \mathbf{Q}_j\xi_j\rp_k$ is the $k$th element of $\mathbf{Q}_j\xi_j$.  Since $\sum_{k=1}^P \lambda_k\lp \mathbf{E}_j^*\mathbf{E}_j\rp = \mathrm{tr}\lp \mathbf{E}_j^*\mathbf{E}_j\rp = P$, there must exist a $k$, say $k^*$, such that $\lambda_{k^*}\lp \mathbf{E}_j^*\mathbf{E}_j\rp \geq 1$. Hence,
\begin{align*}
\| \mathbf{h}_j\|^2 = M\rho \sum_{k=1}^P \lambda_k\lp \mathbf{E}_j^*\mathbf{E}_j\rp \left| \lp \mathbf{Q}_j\xi_j\rp_k\right|^2 \geq M\rho \left| \lp \mathbf{Q}_j\xi_j\rp_{k^*}\right|^2.
\end{align*}
Since $\mathbf{E}_j$ is independent from $\xi_{j}$, and since the distributions of i.i.d. Gaussian vectors are invariant under orthogonal transformations, $\left| \lp \mathbf{Q}_j\xi_j\rp_{k^*}\right|^2$ has the same distribution as $\| \xi_{j,k} \|^2$ for an arbitrary $k$, \emph{i.e.}, $\chi^2(2)$ distribution, or equivalently, exponential distribution with mean $1$. Therefore, the tail probability of $\|\mathbf{h}_j\|^2$ can be lower bounded by $\Prob{\|\mathbf{h}_j\|^2 > M\rho a} \geq e^{-a}$. Hence,
\begin{align*}
\Prob{\left| \mathcal{P}_n\right| \leq (1-\delta)\sqrt{n} } &= \Prob{ \sum_{j=1}^n \mathbb{I}\lp \| \mathbf{h}_j\|^2  \geq M\rho \lp  \frac{1}{2}\log n - 2\log\log n\rp \rp \leq (1-\delta)\sqrt{n} }
\end{align*}
Using the tail lower bound on $\| \mathbf{h}_j\|^2$, we see that each indicator variable is i.i.d. with mean at least $\frac{\log^2 n}{\sqrt{n}}$. Therefore, using Chernoff bound,
\begin{align*}
\Prob{\left| \mathcal{P}_n\right| \leq (1-\delta)\sqrt{n}\log^2 n }\leq O\lp e^{-\delta^2\sqrt{n}\log^2 n}\rp 
\end{align*}
Next, we consider the probability $\Prob{\mathcal{R}_n(1) = \varnothing \left| \left| \mathcal{P}_n\right| \geq (1-\delta)\sqrt{n}\log^2 n \right.}$. Since the users are uniformly distributed in a circle of radius $R$, $\Prob{r_{ij} \leq r} = \frac{r^2}{R^2}$ for sufficiently small $r>0$, and consequently $\Prob{\phi_{ij} \geq x} = \frac{1}{R^2} x^{-\frac{2}{c}}$. Since $h_j$ is independent from $\phi_{1j}$,
\begin{align*}
\Prob{\mathcal{R}_n(1) = \varnothing \left| \left| \mathcal{P}_n\right| \geq (1-\delta)\sqrt{n}\log^2 n \right.} &= \lp 1 - \Prob{\phi_{1j} \geq n^{\frac{c}{4}}}\rp^{(1-\delta)\sqrt{n}\log^2 n} \\
&= \lp 1 - n^{-\frac{1}{2}}\rp^{(1-\delta)\sqrt{n}\log^2 n} = O\lp e^{-(1-\delta)\log^2 n}\rp.
\end{align*}
Then, choosing $\delta= \frac{1}{\log n}$, and by using independence of channels across $i$'s, 
\begin{align*}
\Prob{\mathcal{R}_n(1) \neq \varnothing\;\; \forall i} = \lp 1 - O\lp e^{-(1-\delta)\log^2 n}\rp - O\lp e^{-\delta^2\sqrt{n}\log^2 n}\rp \rp^n = 1 - O\lp e^{-\log^2 n + 2\log n}\rp
\end{align*}
which concludes our proof of the first claim.

To prove the second claim, we note that
\begin{align*}
\| h_i\|^2 &= \rho\left\| \sum_{k=1}^P \xi_{i,k} \mathbf{e}(\theta_{i,k})\right\|^2 \leq \rho  \sum_{k=1}^P \left|\xi_{i,k} \right|^2 \| \mathbf{e}(\theta_{i,k})\|^2 = M\rho X_i,
\end{align*}
where $X_i \sim \chi^2(2P)$. The second claim then follows by Proposition~\ref{prop:min_scaling}.

\section{Proof of Lemma~\ref{lem:main}}\label{ap:main}
Define $\alpha^Q = \sum_{(i,j) \in Q} \sum_{s: \substack{i \in s_1\\ j \in s_2}} \sum_{k \in \mathcal{K}}  \sum_{z \in \mathcal{Z}} \alpha_{skz}$, and consider the following sequence of optimization problems, indexed by $n$ (with a slight abuse of notation):
\begin{align}
\text{maximize } \;\;& U_n\lp \mathbf{\alpha}\rp :=\sum_{i \in \mathcal{N}} U_i \lp \mathbf{\alpha}\rp - \sum_{Q \in \mathcal{Q}} \exp\lbp n\lp \alpha^Q - 1\rp \rbp \label{optn} \\
\text{s.t.} \;\;& \alpha_{skz} \geq 0, \;\; \sum_{s} \alpha_{skz} \leq p_k q_z, \;\; \alpha_{skz} = q_z \sum_{z'} \alpha_{skz'}, \;\; \forall s,k,z. \label{optncons}
\end{align}
We will denote the optimal value of the optimization \eqref{optn} with $\mathsf{OPT}_n$. Further consider the corresponding sequence of scheduling policies $\pi_n$, that choose $s^* = \arg \max_{s \subseteq  \mathcal{N}^2 \times \{ 1,2\} } \wtild f_n(s)$, where
\begin{align}
\wtild f_n(s) = \sum_{(i,j,m) \in s} \mathbb{E} \lb R_{sK(t)Z(t)}^{(i)}\rb \frac{\partial U_i}{\partial r_i}\Bigr|_{\substack{r_i=r_i(t-1)\\ \beta_i=\beta_i (t-1)}} + \frac{\partial U_j}{\partial \beta_j}\Bigr|_{\substack{r_j=r_j(t-1)\\ \beta_j=\beta_j (t-1)}} - n\sum_{Q: s_{12} \cap Q \neq \varnothing}e^{ n \lp \alpha^Q(t) - 1\rp}, \label{policy_n}
\end{align}
The empirical utility of the policy $\pi_n$ up to time $t$ is denoted by $U_n (t)$. 
\begin{prop}\label{prop:opt1}
$\lim_{n \to \infty} \mathsf{OPT}_n = \mathsf{OPT}'.$
\end{prop}
\begin{proof}
We first show that for any $\epsilon>0$, $\mathsf{OPT}_n \geq \mathsf{OPT}' - \epsilon$ for large enough $n$. Consider the optimization \eqref{opt3}, with the condition \eqref{opt3cons2} replaced by
\begin{align}
\alpha^Q \leq 1 + \Delta, \;\forall Q \in \mathcal{Q}, \label{delta_cons}
\end{align}
and denote the optimal value of the resulting maximization as $\mathsf{OPT}^\Delta$. By continuity of the objective function, for any $\epsilon >  0$, there exists $\delta > 0$ such that $\left| \mathsf{OPT}^{-\delta} - \mathsf{OPT}\right| < \frac{\epsilon}{2}$. For such $\delta$, choose $n$ large enough so that $e^{ -n\delta } <  \frac{\epsilon}{2\left| \mathcal{Q}\right|}$. Similarly, denote the maximal value of \eqref{optn} subject to \eqref{delta_cons} as $\mathsf{OPT}_n^{\Delta}$. Then
\begin{align*}
\mathsf{OPT}_n \geq \mathsf{OPT}_n^{-\delta} \geq \mathsf{OPT}^{-\delta} - \frac{\epsilon}{2} \geq \mathsf{OPT}' - \epsilon.
\end{align*}
Next, we show that for large enough $n$, $\mathsf{OPT}_n \leq \mathsf{OPT}' + \epsilon$. Choose $\delta>0$ small enough so that $\left| \mathsf{OPT}^{\delta} - \mathsf{OPT}'\right| < \epsilon$. Hence 
\begin{align*}
\mathsf{OPT}' + \epsilon \geq \mathsf{OPT}^\delta \geq \mathsf{OPT}^\delta_n.
\end{align*}
Therefore it is sufficient to show that $\mathsf{OPT}^\delta_n = \mathsf{OPT}_n$ for large enough $n$. If we choose $n$ large enough so that
\begin{align*}
\frac{\partial U_n \lp \mathbf{\alpha}\rp}{\partial \alpha^Q} \Bigr|_{\alpha^Q>1+\delta} =  \sum_{i \in \mathcal{N}} \frac{\partial U_i \lp \mathbf{\alpha}\rp}{\partial \alpha^Q} - ne^{n\lp \alpha^Q - 1\rp}\Bigr|_{\alpha^Q>1+\delta} < 0,
\end{align*}
then concavity implies $\mathsf{OPT}^\delta_n = \mathsf{OPT}_n$, since the derivative would have to be monotonically decreasing with increasing $\alpha^Q$. Such a choice of $n$ is possible since $\frac{\partial U_i \lp \mathbf{\alpha}\rp}{\partial \alpha^Q}\Bigr|_{\alpha^Q=1+\delta}< \infty$, similarly by concavity and twice continuous differentiability, which concludes the proof. 
\end{proof}

\begin{prop}\label{prop:opt2}
$\lim_{n \to \infty} U_n (t) = U(t)$.
\end{prop}
\begin{proof}
It is sufficient to show that for a given $t$, for a sufficiently large $n$, all the control actions taken by policies $\pi_n$ and $\pi^*$ up to time $t$ are identical. Note that since the sets $\mathcal{K}$, $\mathcal{Z}$ and $\mathcal{N}$ are finite, for a finite $t$, there are finitely many values $\alpha_{skz}(t)$, and therefore $f_s (t)$ can take. Therefore we can choose $n$ large enough so that
\begin{enumerate}
\item For any $\tau \leq t$, if $\alpha^Q(\tau)>1$ for some $Q$, then
\begin{align*}
f_s(\tau) - n\exp\lbp n \lp \alpha^Q - 1\rp\rbp < f_{s^*}(\tau)
\end{align*}
for all subsets $s$ such that $s_{12} \cap Q \neq \varnothing$,
\item For each pair of subsets $s,t \subseteq \mathcal{\bar N}(t) \times \{ 1,2\}$ such that $f_s(\tau) > f_t(\tau)$ and $\alpha^Q (\tau) < 1$ for all $Q$ s.t. $s_{12} \cap Q \neq \varnothing$ and $t_{12} \cap Q \neq \varnothing$,
\begin{align*}
f_s(\tau) - n\sum_{Q: s_{12} \cap Q \neq \varnothing} \exp\{ n\lp \alpha^Q(\tau)-1\rp \} > f_t(\tau) - n\sum_{Q: t_{12} \cap Q \neq \varnothing} \exp\{ n\lp \alpha^Q(\tau)-1\rp \}.
\end{align*}
\end{enumerate}
Here, the first condition ensures that a subset that violates any of the clique constraints is never scheduled, and the second condition ensures that for the subsets whose scheduling does not violate any of the clique constraints, the order with respect to $f$ is preserved, and hence the subset that maximizes $f$ remains the same. This is possible since for $x>0$, $e^{nx}$ can be made arbitrarily large, whereas for $x<0$, it can be made arbitrarily small by scaling $n$. For such $n$, all scheduling decisions of $\pi^*$ and $\pi_n$ up to time $t$ are identical, and thus $U_n (t) = U(t)$ for $n$ sufficiently large.
\end{proof}

\begin{prop}\label{prop:opt3}
$\lim_{t \to \infty} U_n(t) = \mathsf{OPT}_n$.
\end{prop}
\begin{proof}
The proof uses Lyapunov optimization techniques from  \cite{BhattacharyaGeorgiadis_95, TsibonisGeorgiadis_05}. We will make use of the following theorem from \cite{BhattacharyaGeorgiadis_95} to show the result.
\begin{thm}\label{stoch_app}
Consider a stochastic sequence in $\mathbb{R}^p$ satisfying the recursion
\begin{align*}
\mathbf{\alpha}(t) = \mathbf{\alpha}(t-1) + \frac{1}{t}\mathbf{g}(t),
\end{align*}
and let $\{\mathcal{F}_t\}_{t \geq 0}$ be a non-decreasing family of filtrations of the underlying $\sigma$-algebra, such that $\mathbf{g}(t)$ is $\mathcal{F}_t$-measurable.

Assume the following are satisfied.
\begin{enumerate}
\item There exists a compact set $\mathcal{A} \subseteq \mathbb{R}^p$ such that
\begin{align*}
\lim_{t\to\infty} \inf\{ \|\mathbf{\alpha}(t) - \mathbf{\alpha}\|_1: \mathbf{\alpha} \in \mathcal{A}\} = 0,
\end{align*}
\item There exists $K>0$ such that for all $t$, $\| \mathbf{g}(t)\|_1 \leq K$,
\item There exists a twice continuously differentable function $V: \mathbb{R}^p \to \mathbb{R}$ such that
\begin{align*}
\mathbb{E}\lb \mathbf{g}^\top(t+1)|\mathcal{F}_t\rb \nabla V\lp \mathbf{\alpha}(t)\rp < -V\lp \mathbf{\alpha}(t)\rp,
\end{align*}
where $^\top$ represents vector transpose.
\end{enumerate}
Then the function $V$ in condition 3 satisfies $\lim_{t \to \infty} V\lp \alpha(t)\rp^+ = 0$.
\end{thm}

Consider the sequence of vectors $\mathbf{\alpha}(t) = \{\alpha_{skz} (t)\}_{s,k,z}$, whose entries satisfy the recursion
\begin{align*}
\alpha_{skz} (t) = \alpha_{skz} (t-1) + \frac{1}{t}\lp \mathbb{I}_{\mathcal{S}(t)=s}\mathbb{I}_{K(t)=k}\mathbb{I}_{Z(t)=z} - \alpha_{skz} (t-1)\rp.
\end{align*}
Note that the vector $\mathbf{\alpha}(t)$ converges to the compact set defined by \eqref{opt3cons1}--\eqref{opt3cons2}, by the first claim of Theorem~\ref{th:opt}, and the entries of the corresponding update sequence $\mathbf{g}(t)$ in this case is bounded by 1. Following the strategy of \cite{TsibonisGeorgiadis_05}, we choose
\begin{align*}
V(y(t)) &= \sum_{i \in \mathcal{N}} U_i \lp  \sum_{s: i \in s_1} \sum_{k \in \mathcal{K}} \sum_{z \in \mathcal{Z}} R_{skz}^{(i)} \alpha^*_{skz} , \sum_{s: i \in s_2} \sum_{k \in \mathcal{K}} \sum_{z \in \mathcal{Z}} \alpha^*_{skz}\rp - \sum_{Q \in \mathcal{Q}} \exp\lbp n\lp \alpha^{*Q} - 1\rp \rbp \\
& - \sum_{i \in \mathcal{N}} U_i \lp  \sum_{s: i \in s_1} \sum_{k \in \mathcal{K}} \sum_{z \in \mathcal{Z}} R_{skz}^{(i)} \alpha_{skz}(t) , \sum_{s: i \in s_2} \sum_{k \in \mathcal{K}} \sum_{z \in \mathcal{Z}} \alpha_{skz}(t)\rp + \sum_{Q \in \mathcal{Q}} \exp\lbp n\lp \alpha^Q(t) - 1\rp \rbp,
\end{align*}
where $\mathbf{\alpha}^*$ is the solution to \eqref{optn}\footnote{Since \eqref{optn} is the maximization of a continuous function over a compact set, the extreme values are attained within the feasible set.}. Then, if we verify the third condition for this choice of $V$, then the proof is concluded using Theorem~\ref{stoch_app}. 

We first evaluate the terms in the left-hand side of the third condition.
\begin{align*}
&\mathbb{E}\lb g_{skz}^\top(t+1)|\mathcal{F}_t\rb = \mathbb{E}\lb \mathbb{I}_{\mathcal{S}(t+1)=s}\mathbb{I}_{K(t+1)=k}\mathbb{I}_{Z(t+1)=z}|\mathcal{F}_t\rb  - \alpha_{skz} (t) \\
&=\sum_{b \in \mathcal{K}, c \in \mathcal{Z}}\mathbb{E}\lb \mathbb{I}_{\mathcal{S}(t+1)=s}\mathbb{I}_{K(t+1)=k}\mathbb{I}_{Z(t+1)=z}|K(t+1)=b,Z(t+1)=c, \mathcal{F}_t\rb p_b q_c  - \alpha_{skz} (t) \\
&=\mathbb{E}\lb \mathbb{I}_{\mathcal{S}(t+1)=s}|K(t+1)=k,Z(t+1)=z, \mathcal{F}_t\rb p_k q_z  - \alpha_{skz} (t) =\left\{ 
\begin{array}{ll}
p_k q_z  - \alpha_{skz} (t),  & \text{if $s = s^*,$} \\
- \alpha_{skz}(t), & \text{otherwise}
\end{array}
\right.
\end{align*}
where $s^*=\arg \max_{\wtild s \in \mathcal{\bar N}(t+1) \times \{ 1,2\}} \wtild f_n(\wtild s)$. Since a single entry of $\nabla V\lp \mathbf{\alpha}(t)\rp$ is given by
\begin{align*}
D_{skz}:=\frac{\partial V\lp \mathbf{\alpha}(t)\rp}{\partial \alpha_{skz} (t)} = -\sum_{i \in s_1} R_{skz}^{(i)} \frac{\partial U_i}{\partial r_i} \Bigr|_{r_i=r_i(t)} - \sum_{i \in s_2} \frac{\partial U_i}{\partial \beta_i} \Bigr|_{\beta_i = \beta_i(t)} + n\sum_{(i,j) \in s_{12}} \sum_{Q: (i,j) \in Q}e^{ n \lp \alpha^Q(t) - 1\rp},
\end{align*}
and the inner product on the left-hand side of the third condition can be expressed as
\begin{align*}
&\mathbb{E}\lb g^\top(t+1)|\mathcal{F}_t\rb \nabla V\lp \mathbf{\alpha}(t)\rp = -\sum_{k \in \mathcal{K}}\sum_{z \in \mathcal{Z}} D_{s^*kz} p_k q_z +\sum_{k \in \mathcal{K}}\sum_{z \in \mathcal{Z}} \sum_{s} D_{skz} \alpha_{skz}(t) \\  
&= -\sum_{k \in \mathcal{K}} \mathbb{E}_L\lb D_{s^*kL}\rb p_k  +\sum_{k \in \mathcal{K}}\sum_{z \in \mathcal{Z}} \sum_{s} D_{skz} \alpha_{skz}(t) \\  
&\leq -\sum_{k \in \mathcal{K}}\sum_{z \in \mathcal{Z}} \sum_{s} \mathbb{E}_L\lb D_{s^*kL}\rb \alpha_{skz}^*  +\sum_{k \in \mathcal{K}}\sum_{z \in \mathcal{Z}} \sum_{s} D_{skz} \alpha_{skz}(t) \\  
&\leq -\sum_{k \in \mathcal{K}}\sum_{z \in \mathcal{Z}} \sum_{s} \mathbb{E}_L\lb D_{skL}\rb \alpha_{skz}^*  +\sum_{k \in \mathcal{K}}\sum_{z \in \mathcal{Z}} \sum_{s} D_{skz} \alpha_{skz}(t) \\ 
&\leq -\sum_{k \in \mathcal{K}}\sum_{z \in \mathcal{Z}} \sum_{s}\sum_{z' \in \mathcal{Z}}  q_{z'}D_{skz'} \alpha_{skz}^*  +\sum_{k \in \mathcal{K}}\sum_{z \in \mathcal{Z}} \sum_{s} D_{skz} \alpha_{skz}(t) \\
&\overset{\aaaa}{=} -\sum_{k \in \mathcal{K}}\sum_{s}\sum_{z' \in \mathcal{Z}} D_{skz'} \alpha_{skz'}^*  +\sum_{k \in \mathcal{K}}\sum_{z \in \mathcal{Z}} \sum_{s} D_{skz} \alpha_{skz}(t) \\
&=- \sum_{k \in \mathcal{K}}\sum_{s}\sum_{z \in \mathcal{Z}} \frac{\partial V\lp \mathbf{\alpha}(t)\rp}{\partial \alpha_{skz} (t)} \lp \alpha_{skz}^* - \alpha_{skz}(t) \rp \overset{\bbbb}{\leq} -V\lp \mathbf{\alpha}(t)\rp
\end{align*}
where (a) follows by the third constraint in \eqref{optncons}, and (b) follows by convexity.
\end{proof}

Finally, we can prove that $U(t) \to \mathsf{OPT}'$. Note that this is equivalent to the statement
\begin{align*}
\lim_{n \to \infty} \lim_{t \to \infty} U_n (t) = \lim_{t \to \infty} \lim_{n \to \infty} U_n (t).
\end{align*}
Given $\epsilon>0$, using Propositions~\ref{prop:opt1}, \ref{prop:opt2}, and \ref{prop:opt3}, we can find sufficiently large $n$ and $t$ such that
\begin{align*}
\left| U(t) - \mathsf{OPT}\right| \leq \left| U(t) - U_n(t)\right| + \left| U_n(t) - \mathsf{OPT}_n\right| + \left| \mathsf{OPT}_n - \mathsf{OPT}\right| < \frac{\epsilon}{3} + \frac{\epsilon}{3} + \frac{\epsilon}{3} = \epsilon,
\end{align*}
which concludes the proof.

\section{Proof of Theorem~\ref{th:gap}}\label{ap:gap}
\begin{proof}
The upper bound follows by the fact that $R_{MIMO}$ is achievable. To prove the lower bound, we first note that for any input convariance matrix $\mathbf{Q}$,
\begin{align}
\sigma^2_{2|1} = \frac{\left| \Sigma\right|}{\Sigma_{11}} = \frac{\left| \mathbf{I} + \mathbf{HQH}^*\right|}{1 + \| \mathbf{h}_1\|^2} \label{eq:cond_var},
\end{align}
and that $\mathbf{K}^{-1} = \diag\lp 1, \eta\rp$, where $\eta = \frac{1}{1 + \frac{\sigma^2_{2|1}}{\left|g_{12}\right|^2}}$. Next, we lower bound $R_{\text{MIMO}}$ as follows.
\begin{align}
R_{\text{MIMO}} &= \log \left| \mathbf{I}_2 + \mathbf{K}^{-1}\mathbf{HQH}^*\right| \overset{\aaaa}{\geq} \log \left| \mathbf{K}^{-1} + \mathbf{K}^{-1}\mathbf{HQH}^*\right|  \geq  \log \left| \mathbf{I}_2 + \mathbf{HQH}^*\right| + \log \eta, \label{eq:r_low_bd}
\end{align}
To see why (a) holds, define $\mathbf{P}:=\mathbf{K}^{-1}-\mathbf{I}_2$, and denote by $\lambda_k\lp \mathbf{A}\rp$ the $k$'th largest eigenvalue for a matrix $\mathbf{A}$. Then by Weyl's inequality, since $\eta\leq 1$,
\begin{align*}
\lambda_k\lp \mathbf{P}+\mathbf{I}_2 + \mathbf{K}^{-1}\mathbf{HQH}^*\rp &\leq \lambda_k(\mathbf{I}_2 + \mathbf{K}^{-1}\mathbf{HQH}^*) + \lambda_1(\mathbf{P}) = \lambda_k(\mathbf{I}_2 + \mathbf{K}^{-1}\mathbf{HQH}^*),
\end{align*}
which implies latter determinant in \eqref{eq:r_low_bd} is smaller. Next, note that $\eta$ can be lower bounded by
\begin{align}
\eta &\leq \left\{ \begin{array}{ll}
\frac{\left|g_{12}\right|^2}{2\sigma^2_{2|1}} & \text{if $\sigma^2_{2|1} \geq \left|g_{12}\right|^2$} \\
\frac{1}{2} & \text{otherwise}
\end{array} \right. \label{eq:gamma_lb}
\end{align}
Then, combining \eqref{eq:cond_var}, \eqref{eq:r_low_bd}, and \eqref{eq:gamma_lb}, we can show that $R_{\text{MIMO}}$ is lower bounded by
\begin{align*}
R_{\text{MIMO}} \geq \min \lbp \max_{\mathrm{tr}(\mathbf{Q})\leq 1}\log\left| \mathbf{I}_2 + \mathbf{HQH}^*\right|, \log\lp 1 + \| \mathbf{h}_1\|^2\rp + \log^+\lp \left|g_{12}\right|^2\rp\rbp -1,
\end{align*}
where $\log^+(x) := \max\lp 0, \log(x)\rp$. We conclude the proof by noting that for any $x \geq 0$, $\log^+(x) \geq \log(1 + x) -1$, and by the fact that the capacity $\bar C$ is upper bounded by the cut-set bound \cite{CoverElGamal_79}, given by
\begin{align*}
\bar C \leq \min \lbp \max_{\mathrm{tr}(\mathbf{Q})\leq 1}\log\left| \mathbf{I}_2 + \mathbf{HQH}^*\right|, \log\lp 1 + \| \mathbf{h}_1\|^2\rp + \log\lp 1+\left|g_{12}\right|^2\rp\rbp.
\end{align*}
\end{proof}

\section{Proofs of Lemmas~\ref{lem:claim1} and \ref{lem:opt21}}\label{ap:lemmas}
\subsection{Proof of Lemma~\ref{lem:opt21}}
For any $n \in \mathbb{N}$, let $\pi_n$ be a feasible policy such that $\liminf_{t \to \infty} U^{\pi_n}(t) \geq \mathsf{OPT} - \frac{1}{2n}$. Then by definition, there must exist $T_n$ such that for $t>T_n$, $U^{\pi_n}(t) \geq \mathsf{OPT} - \frac{1}{n}$. Consider the sequence $\mathbf{\alpha}^{\pi_n}(T_n)$, where $U^{\pi_n}(t) = U\lp \mathbf{\alpha}^{\pi_n}(t)\rp$. Let the set of vectors $\mathbf{\alpha}$ defined by $\eqref{opt3cons1}$ and $\eqref{opt3cons2}$ be $\mathcal{Y}$. Then strong law of large numbers, and the independence of $\lp \mathcal{S}(t), K(t)\rp$ from $Z(t)$ implies $\lim_{n\to\infty} \inf\lbp \|\mathbf{\alpha} -\mathbf{\alpha}^{\pi_n}(T_n) \|: \mathbf{\alpha} \in \mathcal{Y}\rbp = 0$. Therefore, there exists a sequence $\lbp\mathbf{\alpha}_n\rbp \in \mathcal{Y}$ such that $\lim_{n\to\infty} \|\mathbf{\alpha}_n -\mathbf{\alpha}^{\pi_n}(T_n) \|=0$. Since $\mathcal{Y}$ is closed and bounded, it is compact, and therefore $\mathbf{\alpha}_n$ must have a subsequence, say $\mathbf{\alpha}_{n_k}$, that converges to a point $\mathbf{\alpha}^* \in \mathcal{Y}$, which implies $\lim_{k\to\infty} \mathbf{\alpha}^{\pi_{n_k}}(T_{n_k}) = \mathbf{\alpha}^* \in \mathcal{Y}$.
Since the function $U$ is continuous, we have
\begin{single}
\begin{align*}
\msf{OPT}=\lim_{k \to \infty} U\lp \mathbf{\alpha}^{\pi_{n_k}}(T_{n_k})\rp = U\lp \lim_{k \to \infty} \mathbf{\alpha}^{\pi_{n_k}}(T_{n_k})\rp=U\lp \mathbf{\alpha}^*\rp.
\end{align*}
\end{single}
\begin{double}
\begin{align*}
\msf{OPT}&=\lim_{k \to \infty} U\lp \mathbf{\alpha}^{\pi_{n_k}}(T_{n_k})\rp\\
& = U\lp \lim_{k \to \infty} \mathbf{\alpha}^{\pi_{n_k}}(T_{n_k})\rp=U\lp \mathbf{\alpha}^*\rp.
\end{align*}
\end{double}
Since $\mathbf{\alpha}^*$ is in the feasible set $\mathcal{Y}$, it must be that $\msf{OPT}' \geq U\lp \mathbf{\alpha}^*\rp = \msf{OPT}$.

\subsection{Proof of Lemma~\ref{lem:claim1}}

Assume that there exists $\epsilon>0$, $Q \in \mathcal{Q}$, such that for any $N$, there exists $t>N$ that satisfies $\beta_{Q}^* (t) > 1 + \epsilon$. Note that
\begin{align}
\beta_{Q}^* (t) \leq \frac{t-1}{t}\beta_{Q}^* (t-1) + \frac{\left| Q\right|}{t p} \mathbb{I}_{\beta_{Q}^* (t-1) < 1} \label{eq:pf_ub},
\end{align}
with $p = \min_{(i,j) \in Q} p_{ij}$, where the upper bound is obtained by observing that the maximal increase in $\beta_{Q}^* (t)$ is achieved when all flows $(i,j) \in Q$ are scheduled at slot $t$. Choosing $N = \frac{\left| Q \right|}{\epsilon p}$, there must exist $t>N$ s.t. $\beta_{Q}^* (t) > 1 + \epsilon$. Letting $t^* \geq N$ to be the smallest of such indices, it must be that $\beta_{Q}^* (t^* -1) \leq 1$, since otherwise the increment $\beta_{Q}^* (t) - \beta_{Q}^* (t-1)$ cannot be positive, by construction. But by \eqref{eq:pf_ub} and by the choice of $N$,
\begin{align*}
\beta_{Q}^* (t^*) \leq \frac{t^*-1}{t^*}\beta_{Q}^* (t^*-1) + \epsilon \mathbb{I}_{\beta_{Q}^* (t^*-1) < 1} \leq 1 + \epsilon,
\end{align*}
which is a contradiction. 

\section{Utility Function with Relaying Cost}\label{ap:utility}
For an arbitrary $\kappa$, let $\lp \mathbf{\wtild r}, \wtild \beta\rp$ solve the optimization \eqref{opt2} with $U_i\lp r_i, \beta_i\rp = \log(r_i) + \kappa\log(1 - \beta_i)$,
where
\begin{align*}
r_i = \sum_{s: i \in s_1} \sum_{k \in \mathcal{K}} \sum_{z \in \mathcal{Z}} R_{skz}^{(i)} \alpha_{skz}, \;\;\beta_i = \sum_{s: i \in s_2} \sum_{k \in \mathcal{K}} \sum_{z \in \mathcal{Z}} \alpha_{skz}.
\end{align*}
Note that here $\alpha_{skz}$ has no time dependence and refers to a deterministic quantity, \emph{i.e.}, the fraction of time for which $\mathcal{S}(t)=s, K(t)=k, Z(t)=z$, throughout the (infinite) duration of transmission. Then, for any feasible perturbation $\delta \mathbf{\alpha}$ that pushes the operating point from $\lp \mathbf{\wtild r}, \wtild \beta\rp$ to $\lp \mathbf{r}, \beta\rp$, it must be that $\sum_{s,k,z} \delta \alpha_{skz} \sum_i \frac{\partial U_i}{\partial \alpha_{skz}} \leq 0$ by concavity, which, using the facts 
\begin{single}
\begin{align*}
r_i - \wtild r_i = \delta r_i = \sum_{s: i \in s_1} \sum_{(k,z) \in \mathcal{K}\times\mathcal{Z}} R_{skz}^{(i)} \delta \alpha_{skz},\;\;\;\;\;
\beta_i - \wtild \beta_i = \delta \beta_i = \sum_{s: i \in s_2} \sum_{(k,z) \in \mathcal{K}\times\mathcal{Z}} \delta \alpha_{skz}
\end{align*}
\end{single}
\begin{double}
\begin{align*}
r_i - \wtild r_i = \delta r_i &= \sum_{s: i \in s_1} \sum_{(k,z) \in \mathcal{K}\times\mathcal{Z}} R_{skz}^{(i)} \delta \alpha_{skz},\\
\beta_i - \wtild \beta_i &= \delta \beta_i = \sum_{s: i \in s_2} \sum_{(k,z) \in \mathcal{K}\times\mathcal{Z}} \delta \alpha_{skz},
\end{align*}
\end{double}
can be re-arranged into 
\begin{align}
\sum_i \frac{r_i - \wtild r_i}{\wtild r_i} \leq \kappa\sum_{i} \frac{(1-\wtild \beta_i)-(1-\beta_i)}{1-\wtild \beta_i}. \label{fairness}
\end{align}

\section{Proof of Theorem~\ref{th:ssp}}\label{ap:stability}
Before we present the proof, we need several definitions.
\begin{definition}
The chromatic number $\chi(\mathcal{G})$ is the minimum number of colors needed to color graph $\mathcal{G}$.
\end{definition}
\begin{definition}
The clique number $\omega(\mathcal{G})$ is the maximum clique size in $\mathcal{G}$.
\end{definition}
\begin{definition}
A \emph{perfect graph} is a graph whose chromatic number equals its clique number, \emph{i.e.}, $\chi(\mathcal{G})=\omega(\mathcal{G})$.
\end{definition}
\begin{definition}
A graph is \emph{chordal} if, for every cycle of length larger than three, there is an edge that is not part of the cycle, connecting two of the vertices of the cycle.
\end{definition}
Given these definitions, we are ready for the proof. The results in \cite{TassiulasEphremides_92} can be used to show that the stability region of the constrained queueing network formed by the $n$ users is given by
\begin{align}
\Lambda = \lbp \beta: \mathbf{D}^{-1}\beta \in conv\lp \Pi\rp \rbp, \label{ssp}
\end{align}
where $\mathbf{D}$ is a diagonal matrix with $p_{ij}$ values on the diagonal ($p_{ij}>0$ without loss of generality), $conv(\cdot)$ represents the convex hull of a set of vectors, and $\Pi$ is the set of incidence vectors of the independent sets of $\mathcal{G}_c$, \emph{i.e.}, a vector $\mathbf{s}$ whose elements are indexed by $(i,j)$ is contained in $\Pi$ if $\lbp (i,j): s_{(i,j)} = 1\rbp$ is an independent set of $\mathcal{G}_c$\footnote{The boundary of the stability region is included in the set $\Lambda$ for technical convenience. Note that this does not change the supremum value in the optimization \eqref{opt1} since the objective function is continuous.}.  

The set $\Lambda$ as defined in $\eqref{ssp}$ is known as the stable set polytope of the graph $\mathcal{G}_c$. The exact characterization of $\Lambda$ is not known in general \cite{Rebennack_08}. However, stable set polytopes of perfect graphs can be completely described in terms of their maximal cliques, as characterized in the following theorem.
\begin{theorem}{\cite{Chtaval_75}}\label{th:clique}
Let $\mathcal{Q}$ be the set of maximal cliques of a perfect graph $\mathcal{G}$. Then the stable set polytope of $\mathcal{G}$ is the set of vectors $x \in [0,1]^{| \mathcal{V}|}$ satisfying $\sum_{v \in Q} x_v \leq 1$ for all  $Q\in \mathcal{Q}$.
\end{theorem}
Therefore, to complete the proof, it is sufficient to show that there exists a polynomial-time procedure that adds edges in $\mathcal{G}_c$ such that the resulting graph $\mathcal{\bar G}_c$ is perfect\footnote{The fact that $\Lambda\lp\mathcal{\bar G}_c\rp \subseteq \Lambda\lp\mathcal{G}_c\rp$ follows directly from the fact that $\mathcal{E}_c \subseteq \mathcal{\bar E}_c$}.

It is known that chordal graphs are perfect \cite{Berge_61}, and any graph can be made into a chordal one in polynomial time by inserting edges\footnote{For instance, one can iterate over the vertices, in each iteration connecting all the previously unvisited neighbors of the current vertex to each other. It is easy to show that such a procedure outputs a chordal graph.}. Further, the number of maximal cliques in a chordal graph is upper bounded by the number of nodes (equal to $n(n-1)$ for $\mathcal{G}_c$) \cite{Gavril_74}, and the maximal cliques of a chordal graph can be listed in polynomial time \cite{RosgenStewart_07}, which concludes the proof.

\end{appendices}

\end{document}